\documentclass[12pt]{article}
\usepackage[doublespacing]{setspace}
\usepackage{latexsym}
\usepackage[pdftex]{graphicx}
\usepackage[pdftex]{color}
\usepackage[normalem]{ulem} 
\RequirePackage[OT1]{fontenc}
\RequirePackage{amsthm,amsmath,amssymb}
\RequirePackage{natbib}
\RequirePackage[colorlinks,citecolor=blue,urlcolor=blue]{hyperref}
\usepackage{epsfig,xcolor}
\usepackage{comment}
\newcommand{\mE}{\ensuremath{\mathbb E}}
\newcommand{\mP}{\ensuremath{\mathbb P}}
\newcommand{\mR}{\ensuremath{\mathbb R}}

\newcommand{\E}{\ensuremath{\mathbb E}}
\newcommand{\PP}{\ensuremath{\mathbb P}}

\newcommand{\vZ}{\ensuremath{\vec Z}}
\newcommand{\vz}{\ensuremath{\vec z}}
\newcommand{\vh}{\ensuremath{\vec h}}
\newcommand{\vone}{\ensuremath{\vec 1}}
\newcommand{\vD}{\ensuremath{\vec D}}
\newcommand{\II}{\ensuremath{\mathbb I}}

\newcommand{\tD}{\ensuremath{\tilde{D}}}

\newcommand{\rhcomm}[1]{\textcolor{red}{RH:#1}}

\newcommand{\rhdel}[1]{}

\numberwithin{equation}{section}
\theoremstyle{plain}
\newtheorem{proposition}{Proposition}[section]
\newtheorem{thm}{Theorem}[section]
\newtheorem{lem}{Lemma}[section]

\begin{document}
\noindent {\sffamily\bfseries\Large  Optimal  control of false discovery criteria in the two-group model }

\noindent%

\textsf{Ruth Heller}, Department of
Statistics and Operations Research, Tel-Aviv university, Tel-Aviv  6997801, 
Israel, \textsf{E-mail:} ruheller@gmail.com\\
\textsf{Saharon Rosset}, Department of
Statistics and Operations Research, Tel-Aviv university, Tel-Aviv  6997801, 
Israel, \textsf{E-mail:} saharon@tauex.tau.ac.il\\

\thispagestyle{empty}%

\noindent%

\textsf{Abstract. \
The highly influential two-group model in testing a large number of statistical hypotheses assumes that the test statistics are drawn independently from a mixture of a high probability null distribution and a low probability alternative. Optimal control of the   marginal false discovery rate (mFDR), in the sense that it provides maximal power (expected true discoveries) subject to mFDR control, is known to be achieved by thresholding the  local false discovery rate (locFDR), i.e., the probability of the hypothesis being null given the set of test statistics,  with a fixed threshold. 
We address the challenge of controlling optimally the popular false discovery rate (FDR) or positive FDR (pFDR) rather than mFDR in the general two-group model, which also allows for dependence between the test statistics. These criteria are less conservative than the mFDR criterion, so  they make more rejections in expectation.
We derive their optimal multiple testing (OMT) policies, which  turn out to be thresholding the locFDR with a threshold that is a function of the entire set of statistics.   
We develop an  efficient algorithm for finding these policies, and use  it 
for problems with thousands of hypotheses. 
\rhdel{In simulations with $K=5000$ hypotheses,  the power gain in OMT policies that take the known dependence into account can be very large. For  unknown dependence,  the  misspecified OMT policies which assume the test statistics are independent for FDR or pFDR control  can  have an inflated error level, while  for mFDR control the desired nominal level is maintained.} We illustrate these  procedures on gene expression studies. 
}
\medskip 

\noindent%
\newpage

\setcounter{page}{1}%

\section{Introduction}\label{sec-intro}

In large scale inference problems, hundreds or thousands of hypotheses are tested in order to discover the set of non-null hypotheses.  Such problems are ubiquitous in modern applications in 
medicine, genetics, particle physics,  ecology, and psychology. 
Multiple testing procedures applied to these large scale problems should  control for false discoveries, but they should  not be over-conservative, since this limits the ability of scientists to make true discoveries. Thus it is natural to seek multiple testing procedures that control for false discoveries, while assuring as many discoveries as possible.

In order to guarantee that not too many false positives are among the discoveries, \cite{Benjamini95} introduced the false discovery rate (FDR). This error measure gained tremendous popularity in large scale testing, as it was less stringent than traditional measures like the familywise error rate. 
Given a rejection policy, denote the (random) number of rejected null hypotheses by $R$, and the number of falsely rejected hypotheses (true nulls) by $V$.  The FDR is 
\begin{eqnarray*}
\mbox{FDR: }  \mE\left( \frac{V}{\max(R,1)}\right) = \mE\left( \frac{V}{R}\mid R>0\right)\textrm{Pr} (R>0). 
\end{eqnarray*}

The ``two-group model", first introduced by \cite{Efron01}, has been widely used in large scale inference problems \citep{Efron01,Genovese02, Storey03, Sun07, Efron08, Cai17}. The model assumes that the observed test statistics are generated independently from the mixture model $(1-\pi) g(z\mid h=0)+\pi g(z\mid h=1)$, where $h$ follows the Bernoulli($\pi$) distribution, and indicates whether the null is true (h=0) or the alternative holds (h=1). Accordingly $g(z\mid h=0)$ and $g(z \mid h=1)$ represent the distribution of the test statistic under the null and alternative, respectively. 

In this paper, we consider a more general setting,  in which the test statistics can be dependent  and differ in their marginal distributions.  The more general setting has been used in \cite{Xie11} for short range dependence.  We denote this setting, for which the two-group model is a special case, as the ``general two-group model". As in the two-group model,  the hypotheses states vector, $\vh = (h_1,\ldots, h_K)$, has entries sampled independently from the Bernoulli($\pi$) distribution.
  But in our more general setting, the observed test statistics, $\vec Z = (Z_1,\ldots,Z_k)$, are sampled from the joint distribution given $\vh$, so $\vec Z \mid \vh\sim g(\vz\mid \vh).$


Two measures that are similar to the FDR  became popular within the framework of the two-group model. The pFDR was introduced in \cite{Storey03}. The marginal FDR (mFDR) was introduced in  \cite{Genovese02, Sun07}. Their formulas are: 
\begin{eqnarray*}
\mbox{pFDR: }  \mE\left( \frac{V}{R}\mid R>0\right)
; \quad
\mbox{mFDR: } \frac{\mE V}{\mE R}.
\end{eqnarray*}

For the two-group  model, if the rejection policy is  a fixed subset of the real line, then the pFDR and mFDR have been shown to be equivalent  \citep{Storey03}. Moreover, as $K\rightarrow \infty$, all three measures are equivalent \citep{Benjamini08}. 
 \cite{Cai17} claim that there is essentially no difference between the three measures in large-scale testing problems. They say the use of mFDR is mainly for technical considerations, since the ratio of two expectations is easier to handle. In this paper we  show that for large values of $K$ there can still be important differences when aiming at FDR control
, pFDR control, or mFDR control.

The test statistic that plays a central role for inference on which hypotheses are false is the locFDR, defined for the $i$th hypothesis as $T_i(\vz) = \textrm{Pr}(h_i=0\mid \vz)= \frac{(1-\pi)g(\vz\mid h_i=0)}{(1-\pi)g(\vz\mid h_i=0)+\pi g(\vz\mid h_i=1)}$, where $g(\vz \mid h_i)$ is the joint density of $\vz$ given hypothesis state $h_i$ only, rather than the entire vector $\vh$.  
The locFDR was originally introduced by \cite{Efron01} for standard (i.i.d) two-group model,  where it simplifies since  $\textrm{Pr}(h_i=0\mid \vz)  = \textrm{Pr}(h_i=0\mid z_i)$. We  denote it by $T_{marg}(z)=Pr(h=0\mid z)$ and call it the marginal locFDR.
 \cite{Xie11} showed that for the general two-group model,  the policy which 
minimizes the marginal false non-discovery rate (mFNR), i.e., the 
 expected number of non-null non-rejections divided by the  expected number of non-rejections,
 with mFDR control, is to threshold the locFDR statistics with a fixed threshold under the general dependence setting, following the work of  \cite{Sun07} that showed this for the two-group model. 

In this paper, we consider OMT with FDR or pFDR control in addition to mFDR control for the general two-group model. As \cite{Cai17} have noted, since mFDR is the ratio of two expectations, it is easier to handle when seeking an optimal policy. However,  $V/R$ is arguably  the more fundamental quantity the investigator would like control over for finite $K$. Therefore a rejection policy that guarantees control over $V/R$ in expectation, while maximizing the expected number of true rejections, can be very useful. 
We can write the problem of finding the OMT policy as an optimization problem.  Briefly, 
 for a family of $K$ hypotheses, let $\vec D: \mR^K\rightarrow \{0,1\}^K$ be the decision function when the vector of observed statistics is  $\vec z$, so the $i$th coordinate $D_i(\vec z)$ receives the value of one if the $i$th null hypothesis is rejected, and zero otherwise. Let 
$\vone$ be the vector of ones. Then  the number of rejected and falsely rejected hypotheses, respectively, are $R(\vec D(\vec z)) =\vone^t \vD(\vz)= \sum_{k=1}^K D_k(\vec z)$ and $V(\vec D(\vec z)) = (\vone^t-\vh^t)\vD(\vz) = \sum_{k=1}^K (1-h_k)D_k(\vec z).$ We denote by $Err(\vec D)\in \{FDR(\vec D), pFDR(\vec D), mFDR(\vec D) \} $ the error rate for policy $\vec D$.
We seek to maximize the expected number of true discoveries, $$\Pi(\vec D) = \mE(R(\vec D)-V(\vec D))=\mE \left(\vh^t\vD\right),$$ subject to    $Err(\vec D)\leq \alpha$. 
The solutions to these problems we present  in this paper can be extended
 to optimize other notions of power, e..g., to minimize the expectation of the  loss functions 
$L_\lambda(\vec h, \vec D) = \lambda (\vone^t-\vh^t)\vD + \vh^t(\vone-\vD)$
considered in \cite{Sun07} or the mFNR.  Moreover,  the mathematical and algorithmic developments can be easily adapted for developing the OMT policy for other error measures considered in the literature, such as $\mE(V)$ \citep{Storey07} and  false
discovery exceedance (FDX, where for $\gamma \in (0,1), FDX_{\gamma} = Pr(FDP>\gamma)$, \citealt{LehmannRomano05}). 
The chosen definition should capture the true ``scientific'' goal for inference  and the type of discoveries we wish to make.

OMT problems in the general two-group model can be viewed as an infinite-dimensional optimization problem, seeking to maximize one integral (the power), subject to an integral constraint expressing the measure we want to control ---  FDR, mFDR, pFDR, or any other measure. In this paper we adopt this view and demonstrate that for the two-group model we can solve the resulting optimization problem and practically compute the optimal rejection policy for dimension $K$ in the thousands. 
Our main contributions are as follows. 
\begin{enumerate}
\item We show (Theorem \ref{thm-main}) that the OMT policy for FDR or pFDR control turns out to be thresholding the locFDR with a threshold that is a function of the entire set of statistics. This is true for any dependence structure, and  is in contrast to the previously shown OMT policy for mFDR control \citep{Sun07, Xie11}, where the threshold is fixed. 

\item We provide efficient algorithms for finding the OMT policies (\S~\ref{sec-OMT-FDR-pFDR} and \S~\ref{sec-algorithm}), which are based on our formulation of the problem as an infinite integer optimization problem with a single constraint. For dependent test statistics, we address the additional computational challenge of computing the locFDR values.  The infinite-dimensional formulation to finding OMT policies under frequentist strong control of measures like family-wise error rate (FWER) and FDR was applied recently in \cite{Rosset18}. In that setting, it was possible to solve only relatively small problems, and only under exchangeability assumptions, due to the complex constraint structure of strong control. In contrast, the structure of the two-group model, with a single constraint, and the computational shortcuts we introduce below, allow us to find OMT policies for practically any $K$.

\item  Via numerical investigations (\S~\ref{sec-sim}), we show that: there is a  huge potential power gain from incorporating known dependence into the OMT procedures; the power increase from controlling FDR and pFDR over mFDR can be  non-negligible even for thousands of hypotheses; when the signal is weak,  the OMT policy with pFDR control has a significantly lower   probability of zero rejections than  the OMT policy with FDR control, and  the pFDR is (arguably) preferred over FDR  to control optimally for the two-group model. We also demonstrate the potential usefulness of the methods to  gene expression studies (\S~\ref{sec-gene-expression}).


\end{enumerate}


\section{Properties of OMT policies}



For independent test statistics from the two-group model,   \cite{Sun07} showed that 
the OMT policy for mFDR control  is to reject the hypotheses with 
$T_{marg}(z) \leq t_{\alpha}$,  where 
 \begin{equation}\label{eq-talpha}
t_{\alpha} = \max \left \lbrace t: \frac{\sum_{k=1}^K\mE\lbrace(1-h_k) \II(T_{marg}(z_k) \leq t)\rbrace}{\sum_{k=1}^K\mE\lbrace \II(T_{marg}(z_k) \leq t\rbrace} \leq \alpha \right \rbrace.
\end{equation}
 
Since  $\sum_{k=1}^K\mE\lbrace(1-h_k) \II(T_{marg}(z_k) \leq t)\rbrace = K\times \textrm{Pr} (h=0, T_{marg}(z) \leq t)$ and $\sum_{k=1}^K\mE\lbrace \II(T_{marg}(z_k) \leq t)\rbrace = K\times \textrm{Pr} ( T_{marg}(z) \leq t)$,  it follows that 

 $$
  \frac{\sum_{k=1}^K\mE\lbrace(1-h_k) \II(T_{marg}(z_k) \leq t)\rbrace}{\sum_{k=1}^K\mE\lbrace \II(T_{marg}(z_k) \leq t\rbrace}  = \textrm{Pr} (h=0\mid  T_{marg}(z) \leq t).
  $$
\cite{Storey03} used this observation to show that when the rejection policy is a fixed region of the real line, $pFDR = mFDR$. Therefore, the optimal rule has a nice Bayesian interpretation: by reporting a hypothesis as non-null if $T_{marg}(z)\leq t_{\alpha}$,  then the mFDR is the chance that a false discovery was made, since 
$ mFDR  = pFDR = \textrm{Pr}(h=0 \mid T(z)\leq t_{\alpha}).$
Since $FDR = pFDR\times Pr(R>0)\leq pFDR$, and  since the OMT policy with mFDR control is a fixed region of the real line, it follows that the OMT policy with mFDR control at level $\alpha$ also controls the FDR and the pFDR at level $\alpha$. Therefore, the OMT policy  with FDR or pFDR  control  will necessarily be at least as powerful as the OMT policy  with mFDR control. Interestingly, for any  number of hypotheses $K$ from the two-group model, the OMT policies with pFDR or FDR control necessarily differ from the OMT policy with mFDR control. This difference exists  even if  the probability of zero rejections is one, and will become clear from the algorithm for constructing the OMT policies with pFDR and FDR control in \S~\ref{sec-algorithm}. 
We formalize these interesting properties and few others in the following proposition. 

\begin{proposition}\label{prop1}
For $K$ test statistics independently drawn from the two-group model
, if the null and non-null distributions  have  positive densities with respect to the Lebesgue measure on their region of support, then: 
\begin{enumerate}
\item $\Pi_{OMT-FDR}\geq \Pi_{OMT-pFDR} \geq \Pi_{OMT-mFDR}$, where $\Pi_{OMT-Err}$ is  the expected number of true discoveries, for the OMT policy with level $\alpha$ control of $Err\in \{ FDR, pFDR, mFDR\}$.
\item The OMT policy with mFDR control differs from the OMT policy with pFDR control. 
\item $mFDR> pFDR$ for the OMT policy with pFDR control. 
\item If the OMT policy with FDR control has probability zero of no rejections, then it coincides with the OMT policy with pFDR control.  
\end{enumerate}
\end{proposition}
See Appendix \ref{app-proofs} for the proof.

Interestingly, the  mFDR controlling procedure described above  controls the mFDR at the nominal level even if the test statistics are dependent or differ in their marginal distributions, as expressed in the following simple result.

\begin{proposition}\label{prop2}
For $K$ test statistics  drawn from the general two-group model, the procedure that rejects the hypotheses with 
$T_{marg}(z_i) \leq t_{\alpha}$, $i=1, \ldots, K$,  satisfies $mFDR \leq \alpha$. 
\end{proposition}
\begin{proof}
For a single coordinate in the vector of test statistics from the general two-group model,  
$$
 \mE_{h,Z}\lbrace (1-h)\II(T_{marg}(Z)\leq t)\rbrace  = \mE_Z [\mE\lbrace (1-h)\mid Z \rbrace \II(T_{marg}(Z)\leq t)]= \mE_Z \lbrace T_{marg}(Z)\II(T_{marg}(Z)\leq t)\rbrace. 
$$
So all the expectations in \eqref{eq-talpha} depend only on the marginal distributions of the test statistics, and therefore
 $t_{\alpha}$ depends only on the marginal distributions of  $Z_1, \ldots, Z_k$. 
The $mFDR$ of the suggested procedure for test statistics drawn from the general two-group model is $ \frac{\sum_{k=1}^K\mE\lbrace T_{marg}(Z_k)\II(T_{marg}(Z)\leq t)\rbrace}{\sum_{k=1}^K\mE\lbrace \II(T_{marg}(Z_k) \leq t_{\alpha}\rbrace}$, so it follows directly from the definition of $t_{\alpha}$ that this mFDR is bounded above by $\alpha$. 
\end{proof}

If  the  test  statistics from the general two-group model are  dependent, thresholding the {\em marginal} locFDRs is not optimal for mFDR control. However, \cite{Xie11} showed that the OMT policy for mFDR control  is  of the form $T_i(\vz)\leq t$ for $T_i(\vz)=Pr(h_i=0 |\vz)$ the true locFDR. We provide an alternative proof in   \S~\ref{app-mfdr-rule}. Our first key theoretical result is that the optimal FDR and pFDR controlling procedures also reject the hypotheses by thresholding the locFDRs.

\begin{thm}\label{thm-main}
Let $\vz$ be a vector of $K$ test statistics coming from the general two-group model. Then for $Err \in \{FDR, pFDR \}$, the OMT decision policy which satisfies $Err(\vD^{Err})\leq \alpha$ and $\mE(\vh^t \vD^{Err})\geq  \mE(\vh^t \vD) \quad \forall \vD \ \textrm{s.t.} \ Err(\vD)\leq \alpha $, is 
almost surely 
{\em weakly monotone} in the locFDR values: 
$$ T_i(\vz) \geq T_j(\vz)  \Leftrightarrow D^{Err}_i(\vz) \leq D^{Err}_j(\vz).$$
\end{thm}


The proof of this theorem is also given in Appendix  \ref{app-proofs}. 

In the next section we shall show that the threshold for rejection with optimal FDR or pFDR control depends on all the realized statistics. This is in contrast to the threshold for optimal mFDR control, which is fixed. More specifically, we shall show that the OMT policies with FDR or pFDR control are step-down procedures , in contrast with the single step procedure for controlling the mFDR optimally (see, e.g.,  \citealt{LehmannRomano05}, for the distinction between step-down, step-up, and single step procedures).

\section{Optimal procedures for FDR or pFDR control in the general two-group model
}\label{sec-OMT-FDR-pFDR}

Given the  selected power measure, the expected number of true positive findings, and false discovery measure to control
$Err\in \{FDR, pFDR \}$, we can write the OMT problem as an infinite dimensional integer program,

\begin{eqnarray}
\max_{\vD : \mR^K \rightarrow \{0,1\}^K}  && \E(\vh^t \vD) \label{prob}\\
\mbox{s.t. }&& Err(\vD)  \leq \alpha \label{const}.
\end{eqnarray}

Let $\PP(\vz) = \sum_{\vh}g(\vz\mid \vh)\pi^{\vone^t \vh}(1-\pi)^{K-\vone^t \vh}$ denote the joint distribution of the test statistics.

The objective is linear in $\vec D$: 
\begin{eqnarray}
\E(\vh^t \vD) && = \int_{\mR^K} \sum_{\vh}\left\lbrace \sum_{i=1}^K D_i(\vz) h_ig(\vz\mid \vh)\pi^{\vone^t \vh}(1-\pi)^{K-\vone^t \vh}\right\rbrace  d\vz \nonumber\\
&& = \int_{\mR^K} \sum_{i=1}^K D_i(\vz)\left\lbrace \sum_{\vh} h_i g(\vz\mid \vh)\pi^{\vone^t \vh}(1-\pi)^{K-\vone^t \vh} \right\rbrace d\vz \nonumber\\
&& = \int_{\mR^K} \sum_{i=1}^K D_i(\vz) \lbrace 1-T_i(\vz) \rbrace \PP(\vz) d\vz \nonumber
\end{eqnarray}
where the last inequality follows since $\sum_{\vh} h_i g(\vz\mid \vh)\pi^{\vone^t \vh}(1-\pi)^{K-\vone^t \vh} = \textrm{Pr} (h_i=1 \mid \vz) \PP(\vz)$ and  $ \textrm{Pr} (h_i=1 \mid \vz)=1-T_i(\vz)$ from the locFDR definition.


The constraint can also be expressed in terms of the locFDR values and $\PP(\vz)$:
\begin{eqnarray}
&&FDR(\vD) =  \int_{\mR^K} \sum_{\vh} \frac{(\vone^t-\vh^t) \vD(\vz)} {{\vone}^t \vD(\vz)} g(\vz\mid \vh)\pi^{\vone^t \vh}(1-\pi)^{K-\vone^t \vh}   d\vz \nonumber \\
&& =
\int_{\mR^K}  \sum_{i=1}^K\frac{D_i(\vz)} {{\vone}^t \vD(\vz)}\sum_{\vh} (1-h_i) g(\vz\mid \vh)\pi^{\vone^t \vh}(1-\pi)^{K-\vone^t \vh}   d\vz
\nonumber \\
&& =\int_{\mR^K}  \sum_{i=1}^K\frac{D_i(\vz)} {{\vone}^t \vD(\vz)}T_i(\vz)  \PP(\vz)    d\vz
\leq \alpha, \label{const-FDR} \\
&& pFDR(\vD) = \frac{FDR(\vD)}{\int_{\mR^K}  \II \{\vec 1 ^t \vD(\vz)>0 \}  \PP(\vz) d\vz} \label{const-pFDR}\leq \alpha,
\end{eqnarray}
To simplify the notation, we employ in our FDR calculations the convention $0/0=0$.

Denote by $\vD^*$ an optimal solution of this problem. 
As written, this is an  infinite integer program, with an objective that is linear in $\vD$ but a constraint which is a non-linear function of $\vD$. In this section, we prove that:
\begin{enumerate}
\item The optimal solution has a structure which allows us to write the constraint as a linear functional of $\vD$ (using Thm. \ref{thm-main} above).
\item Once the problem is written in this linear fashion, the infinite linear program relaxation of the infinite integer problem is guaranteed to have a solution that is integer almost everywhere 
(Lemma  \ref{lem2}).
\item This infinite linear program is guaranteed to have zero duality gap, and hence its solution can be found by solving the Euler-Lagrange conditions, and a solution to these can be found via one-dimensional search (Lemma  \ref{lem3}). 
\end{enumerate} 
Taken together, these results establish a practical methodology to solve the general two-group FDR or pFDR control problem. 
In the next section, we discuss the algorithmic and computational aspects, establishing that this problem can be practically solved for high dimensional settings for the i.i.d two-group model, and in some important cases also for general two-group settings with dependence, yielding the optimal FDR or pFDR controlling policy. 
 

We note that Lemmas  \ref{lem2} and \ref{lem3} are similar in nature, and employ similar techniques, to results in our previous work on multiple testing under strong control \citep{Rosset18}, although some of the important details differ. A major  difference in details is  the fact  our decision rule is not necessarily symmetric in $\vz$ when the data is dependent. Another important distinction is that  
the infinite linear program  in this work has only a single error constraint and can be solved practically for large  $K$, whereas in \cite{Rosset18} it has $K$ error constraints and  can be solved only for a very low dimension $K$.





Theorem \ref{thm-main} demonstrates that for every $\vz$ the optimal policy rejects the set of hypotheses with the smallest locFDR, up to a threshold: $D^*_i(\vz) = 1 \Leftrightarrow T_i(\vz) \leq t(\vz).$

With this characterization of the optimal solution, we can rewrite the constraint in Problem (\ref{prob},\ref{const}) so it is linear in $\vD$. To simplify notation, we replace $\vD$ with a version $\tD$ which operates on the sorted locFDR in increasing order. Explicitly, given $\vz,$ let $i_1,\ldots,i_K$ be the sorting permutation, so $T_{i_1}(\vz)\leq T_{i_2}(\vz)\leq\ldots\leq T_{i_K}(\vz),$ then we define $\tD_k(\vz) = D_{i_k}(\vz).$ Given the characterization of $D^*$ above, then for every $\vz$ we can find $k^*(\vz)$ such that $ \tD^*_k (\vz) = 1 \Leftrightarrow k\leq k^*(\vz).$

We can therefore write:
\begin{equation}\label{const-lin-FDR}
FDR(\tD) =  \int_{\mR^K} \PP(\vz)\left[\tD_1(\vz)T_{(1)}(\vz)+ \sum_{k=2}^K \tD_i(\vz)  \frac{1}{k} \left(T_{(k)}(\vz) - \bar{T}_{k-1} (\vz) \right)\right]d\vz, 
\end{equation}
where  $\bar{T}_{k-1}(\vz)=\frac{\sum_{l=1}^{k-1}  T_{(l)}(\vz)}{k-1}$, and $T_{(k)}(\vz)$ denotes order statistic, i.e. the $k$th smallest locFDR value.
 See Appendix \ref{app-derFDR} for details of the  derivation of the formulation (\ref{const-lin-FDR}).  Using this representation, the pFDR constraint in (\ref{const-pFDR}) also has a linear representation:
\begin{equation}\label{const-lin-pFDR}
 FDR(\tD) - \textrm{Pr}(R>0)\alpha = FDR(\tD)- \int_{\mR^K} \PP(\vz)\tD_1(\vz) \alpha d\vz \leq 0
\end{equation}

To emphasize the linearity of the objective and constraints, and simplify the followup, we rewrite our formulation in a generic form: 
\begin{eqnarray}\label{prob-genericform}
\max_{\tD : \mR^K \rightarrow \{0,1\}^K}  && \int_{\mR^K} \PP(\vz) \sum_{k=1}^K \tD_k(\vz) a_k(\vz)  d\vz \label{prob1}\\
\mbox{s.t. }&& \int_{\mR^K} \PP(\vz) \sum_{k=1}^K \tD_k(\vz) b_k(\vz) d\vz \leq c_{Err}, \nonumber\\
&&\tD_1(\vz)\geq \tD_2(\vz)\geq \ldots \geq \tD_K(\vz),\; \forall \vz \in \mR^K \nonumber, 
\end{eqnarray}
where $a_k, b_k, k=1,\ldots, K$ are functions of 
the locFDR order statistics, and $c_{Err}$ is a fixed constant. Specifically, for $Err(\tD) = FDR(\tD)$: $a_k(\vz)=1-T_{(k)}(\vz), k=1, \ldots,K$; 
 $b_k(\vz)= \left(T_{(k)}(\vz) - \bar{T}_{k-1} (\vz) \right)/k, k=2,\ldots,K$; $b_1(\vz) = T_{(1)}(\vz)$; $c_{Err} =c_{FDR}= \alpha$. For $Err(\tD) = pFDR(\tD)$, the only differences are that $b_1(\vz) = T_{(1)}(\vz)-\alpha$ and  $c_{Err} =c_{pFDR}= 0$.

We now consider the relaxed linear program without the integer requirement on $\tD$, by writing the same problem, except optimizing over $\tD(\vz) \in [0,1]^K$:
\begin{eqnarray}
\max_{\vD : \mR^K \rightarrow [0,1]^K}  && \int_{\mR^K} \PP(\vz) \sum_{k=1}^K \tD_k(\vz) a_k(\vz)  d\vz \label{prob-lin}\\
\mbox{s.t. }&& \int_{\mR^K} \PP(\vz) \sum_{k=1}^K \tD_k(\vz) b_k(\vz) d\vz \leq c_{Err} \nonumber\\
&&\tD_1(\vz)\geq \tD_2(\vz)\geq \ldots \geq \tD_K(\vz),\; \forall \vz \in \mR^K \nonumber. 
\end{eqnarray}

To analyze this problem, we consider its Euler-Lagrange (EL) necessary optimality conditions \citep{korn2000mathematical}. 
We derive the EL conditions for this problem in Appendix \ref{app-proofs}, and also show there that they can be rephrased as requiring the following to hold almost everywhere for optimality, in addition to the (primal feasibility) constraints of Problem (\ref{prob-lin}):
\begin{eqnarray}
&&a_k(\vz) - \mu b_{k}(\vz)-\lambda_{k}(\vz)+\lambda_{k+1}(\vz) = 0,\;\;\forall \vz \in \mR^K, k=1,\ldots,K. \label{KKT-stat}\\
&& \mu\left\lbrace \int_{\mR^K} \left(\sum_{k=1}^K b_{k}(\vz)\tD_k(\vz)\right) \PP(\vz) d\vz-\alpha\right\rbrace=0,\label{KKT-CS1} \\
&&\lambda_{K+1}(\vz)\tD_K(\vz) = 0 \ \forall \vz\in \mR^K  \label{KKT-CS2}\\
&&\lambda_j(\vz) (\tD_{j -1}(\vz)-\tD_j(\vz)) = 0\;,\;\forall \vz \in \mR^K,\;j=2,\ldots,K \label{KKT-CS3}\\
&&\lambda_1(\vz) (\tD_1(\vz)-1 ) = 0\;,\;\forall \vz \in \mR^K\label{KKT-CS4},
\end{eqnarray}
where $\mu$ and $\lambda_j(\vz),\;j=1,\ldots,K+1,\;\vz\in \mR^K$ are non-negative Lagrange multiplies. In analogy to the Karush-Kuhn-Tucker (KKT) conditions in finite convex optimization, we can term condition (\ref{KKT-stat}) the {\em stationarity} condition, and conditions (\ref{KKT-CS1}--\ref{KKT-CS4}) the {\em complementary slackness} conditions.

The following result clarifies that for this problem, we can solve the linear program relaxation instead of the integer program, and get an integer solution:
\begin{lem} \label{lem2}

\newcommand{\vv}{\ensuremath{\vec v}}
For $K$ test statistics drawn from the general two-group model, assume the following non-redundancy condition:
\begin{equation}
\forall \vv\in \mR^K,\; \vv\neq 0,\;\;\mP\left(\sum_k v_k T_k(\vec Z) = 0 \right) = 0.\label{eq:nonred}
\end{equation}
Then any solution to the EL conditions \eqref{KKT-stat}--\eqref{KKT-CS4}  is integer almost everywhere on $\mR^K$. 
\end{lem}

Note that the non-redundancy condition Eq. (\ref{eq:nonred}) is very mild, as it is satisfied whenever the distribution of $\vec Z$ is continuous and $T_k(\vec Z)$ are non-linear functions, which is the case in all standard applications.
   
Our next result shows that for our problem, the EL conditions are in fact not only necessary, but also sufficient (like the KKT conditions in finite linear programs), and we can thus find the infinite linear program solution by finding any solution that complies with these conditions.
\begin{lem} \label{lem3}
The infinite linear program (\ref{prob-lin}) has zero duality gap, and therefore the conditions \eqref{KKT-stat}--\eqref{KKT-CS4} together with primal feasibility are also sufficient, and a solution complying with these conditions is optimal.
\end{lem}
For brevity, we defer explicit derivation of the dual together with the proof to Appendix \ref{app-proofs}. 

Putting our results together, we obtain our an explicit characterization of the OMT solution to our problems of interest: 
\begin{thm} 
For $K$ test statistics drawn from the general two-group model,
an optimal solution to Problem (\ref{prob},\ref{const}) can be found by solving the EL conditions  \eqref{KKT-stat}--\eqref{KKT-CS4} together with primal feasibility of the infinite linear program (\ref{prob-lin}).
\end{thm}

We next show how this can be used to efficiently solve high-dimensional multiple testing problem with FDR or pFDR control for the two-group model.

\section{Algorithm}\label{sec-algorithm}
We first characterize a generic algorithm 
to solve the OMT problem with FDR or pFDR control. We then show how to 
efficiently implement this approach for high dimensional instances of the problem.

Given a candidate Lagrange multiplier $\mu\geq 0$, and an efficient method for calculating locFDR values $T_k(\vz)$, for $k=1,\ldots,K$, define: 
$ R_k(\vz) = a_k(\vz) - \mu b_{k}(\vz) $. For $a_k(\vz)$ and $b_k(\vz)$  defined for the FDR and pFDR constraints, $ R_k(\vz)$ is as follows: 
\begin{eqnarray*}
&&R_1(\vz) = \begin{cases}
 1 - T_{(1)}(\vz) - \mu T_{(1)}(\vz) & \text{if } Err = FDR,\\
1 - T_{(1)}(\vz) - \mu (T_{(1)}(\vz)-\alpha) & \text{if } Err = pFDR.\\
\end{cases}\\
&& R_k(\vz) = a_k(\vz) - \mu b_{k}(\vz) = 1 - T_{(k)}(\vz) - \frac{\mu}{k} (T_{(k)}(\vz) - \bar{T}_{k-1}(\vz)) \text{ for } k=2,\ldots,K.
\end{eqnarray*}

Denote by $\tD^\mu(\vz)$ a solution which complies with \eqref{KKT-stat} and  \eqref{KKT-CS2}--\eqref{KKT-CS4} for this value of $\mu$. It is easy to confirm that this dictates that almost surely:
\begin{eqnarray}
\tD^\mu_1(\vz) &=&  \II\left\{\cup_{l=1}^K \left(\sum_{k=1}^l R_k(\vz) > 0\right)\right\} \label{Dmu1}\\
\tD^\mu_i(\vz) &=&  \II\left\{\left(\tD^\mu_{i-1}(\vz)=1\right) \cap \cup_{l=i}^K \left(\sum_{k=i}^l R_k(\vz) > 0\right)\right\},\; i=2,\ldots,K. \label{Dmu2}
\end{eqnarray}
 
Now we have to ensure that primal feasibility and complementary slackness for $\mu$ hold, in other words find $\mu^* \geq 0$ such that the following holds:
\begin{equation} 
\int_{\mR^K} \PP(\vz) \sum_{k=1}^K \tD_k(\vz) b_k(\vz) d\vz = c_{Err}. \label{mu-cond}
\end{equation} 
It is easy to confirm that if we find such a solution, then it is feasible, it complies with conditions \eqref{KKT-stat}--\eqref{KKT-CS4}, and it is obviously binary.
Thus, finding the optimal solution amounts to searching the one-dimensional space of $\mu$ values for a solution of Eq. \eqref{mu-cond}, using the characterization in Eqs. \eqref{Dmu1}, \eqref{Dmu2}. 

When naively implemented, the calculation in Eqs. \eqref{Dmu1},\eqref{Dmu2} requires $O(K^2)$ operations to calculate all partial sums. However we can rephrase it using a recursive representation to require only $O(K)$ calculations. We first calculate, in decreasing order: 
\begin{equation*}
m_{K}(\vz) = \max(0,R_K(\vz))\;,\;\;\;
m_k(\vz)  =  \max\left(0, m_{k+1} + R_k(\vz)\right) \;,\;k=K-1,\ldots,1,
\end{equation*}
and then, in increasing order: 
\begin{equation*}
\tD^\mu_1 = \II \left\{m_1 > 0\right\}\;,\;\;\;\tD^\mu_k = \II \left\{\left(\tD^\mu_{k-1}=1\right) \cap m_k > 0\right\}\;,\;k=2,...,K.
\end{equation*}

We see from the algorithm that the OMT procedure with FDR control starts by determining whether the hypothesis with the smallest locFDR can be rejected, and proceeds to decide whether to reject the hypothesis with the second smallest locFDR only if the decision at the first step was to reject (i.e., $D_1^\mu=1$). Proceeding similarly, only if the hypothesis with the $l$th smallest locFDR is rejected, the hypothesis with the $(l+1)$th smallest locFDR is tested, for $l=1,\ldots,K-1$. Thus, it is a step-down procedure \citep{LehmannRomano05}. In contrast, the OMT procedure with mFDR control is a single step procedure since each hypothesis is rejected if its locFDR is less than a common cut-off value.   

Implementing the algorithm allows us to find optimal solutions to two-group FDR problems with many thousands of hypotheses in minutes of CPU, as illustrated in \S~\ref{sec-sim} and \S~\ref{sec-gene-expression}.

\subsection{Calculating locFDR values under dependence}

We first note that under the standard two-group model with i.i.d assumptions, calculating $T_i(\vz),\;i=1,\ldots,K$ requires $O(K)$ calculations, and thus does not increase the complexity of the algorithm above. 

Under general (known) dependence, the calculation of the locFDR involves calculating the terms $g(\vz\mid h_i=0)$ and $g(\vz\mid h_i=1)$ (or simply $g(\vz)$). A naive calculation requires integrating over all $O(2^K)$ possible allocations of the vector $\vh,$ for example: 
$$ g(\vz\mid h_i=0) = \sum_{\vh\in \{0,1\}^K:h_i=0} \pi^{\vone^t\vh}(1-\pi)^{K-\vone^t\vh-1} g(\vz \mid \vh),$$ 
where $g(\vz \mid \vh)$ is the known joint distribution of the test statistics under the configuration $\vh.$ Even assuming the calculation of $g(\vz \mid \vh)$ itself is easy, the summation makes this impossible for large $K.$

We discuss two dependence structures where it is possible to design more efficient algorithms:
\begin{itemize}
    \item {\bf Block dependence:} In this setting, the hypotheses $1,\ldots,K$ are partitioned into $1<L<K$ blocks, with $Z_i \perp Z_j$ if $i,j$ do not belong to the same block. As we show below, in the setting we can calculate the set of locFDRs in complexity that depends exponentially on the size of the biggest block, and only linearly on $K.$ This block structure is often assumed in various applications (REFs), and is the example we study in detail in \S~ \ref{subsec-simdep} below.
    \item {\bf Equi-correlated setting:} This refers to the specific setting where the distributions under both the null and alternative are normal with the same variance, and all hypotheses are dependent with an equal correlation for all pairs. In this setting we can design a highly efficient specialized algorithm that requires only $O(K^3)$ operations to calculate all locFDRs. We present it in \S~\ref{app-equicor}.  
\end{itemize}

For the block dependence setting, assume we have a partition into blocks $B_1,\ldots B_L$ such that $\cup_{l=1}^L B_L = \left\{1,\ldots,K\right\}\;,\;\; B_l \cap B_m = \emptyset\;,\;l\neq m,$
and denote the size of block $l$ by $s_l=|B_l|.$
Then it is easy to see that the joint distributions we are interested in factor due to independence, for example, assume WLOG $i\in B_1,$ then we have
$$ g(\vz\mid h_i=0) = \sum_{\vh_1\in \{0,1\}^{s_1}:h_i=0} \pi^{\vone_1^t\vh_1}(1-\pi)^{s_1-\vone_1^t\vh_1-1} g(\vz_1\mid\vh_1) \cdot \prod_{l=2}^L \left[
\sum_{\vh_{l}\in \{0,1\}^{s_l}} \pi^{\vone_{l}^t\vh_{l}}(1-\pi)^{s_l-\vone_{l}^t\vh_{l}}g(\vz_l\mid\vh_l)\right].$$

The same product as in the above display also appears in $ g(\vz\mid h_i=1)$ and $g(\vz),$ so we can combine them to show that the calculation of the locFDR depends only on its dependence block (still assuming $i\in B_1$):
$$T_i(\vz) = \frac{\sum_{\vh_1\in \{0,1\}^{s_1}:h_i=0} \pi^{\vone_1^t\vh_1}(1-\pi)^{s_1-\vone_1^t\vh_1} g(\vz_1\mid\vh_1)}{\sum_{\vh_1\in \{0,1\}^{s_1}} \pi^{\vone_1^t\vh_1}(1-\pi)^{s_1-\vone_1^t\vh_1} g(\vz_1 \mid\vh_1)}.$$
The denominator can clearly be calculated via $2^{s_1}$ evaluations of $g(\vz_1\mid\vh_1),$ and is fixed for all $i\in B_1.$ A naive calculation of the numerator for all $i\in B_1$ requires $s_1 \times 2^{s_1}$ evaluations. We note that the evaluations can be done only once for each $\vh$ and stored with $O(2^{s_1})$ memory.

Overall, assuming the evaluation of $g(\vz_1\mid\vh_1)$ for a block of size $s_1$ is of complexity $O(s_1^2)$ (as for a multivariate Gaussian with known covariance structure), the total complexity of calculating all locFDRs in a block design with $L$ blocks, each of size $s=K/L$ is $O(K \cdot s \cdot 2^s).$

\section{Numerical Examples}\label{sec-sim}

We compare the performance of the OMT procedure with FDR control (henceforth, OMT-FDR) and the OMT procedure with positive FDR control (henceforth, OMT-pFDR), against two natural competitors: the OMT procedure with mFDR control (henceforth, OMT-mFDR, \citealt{Xie11}), and the oracle BH procedure, which applies the BH procedure assuming the probability of a null hypothesis is known (so the threshold for significance of the $i$th largest $p$-value is $\frac{i\alpha}{K(1-\pi)}$ instead of the BH threshold $\frac{i\alpha}{K}$, \citealt{Benjamini06}). In \S~\ref{subsec-simind} we examine the case that the test statistics are independent; in \S~\ref{subsec-simdep} we examine dependence settings, in which case we also compare the OMT procedures with the misspecified procedures that find the OMT policies assuming the test statistics are independent (termed ind-Err when aimed at Err control, where Err is FDR, pFDR,  or mFDR); in \S~\ref{subsec-simest} we examine the effect of estimating the mixture parameters from the data.  In \S~\ref{subsec-simdep} and \S~\ref{subsec-simest} we also compare with the BH procedure and with the adaptive procedure suggested in \cite{Sun07}, which is computationally simpler and more intuitive than OMT-mFDR, and therefore quite popular   for large scale inference. This procedure, termed here est-mFDR, first orders the estimated marginal locFDRs, $T_{marg,(1)}\leq \ldots \leq T_{marg,(K)}$, and then rejects the $k$ hypotheses with smallest estimated marginal locFDRs, where $k = \max \lbrace i: \frac 1i \sum_{j=1}^i T_{marg,(j)}\leq \alpha \rbrace$.  
All simulations are carried out at the nominal level $\alpha = 0.05$ for the chosen criterion (mFDR, FDR or pFDR).

\subsection{The independent setting}\label{subsec-simind} 
We generate test statistics from the following mixture model: with probability $1-\pi$, $Z$ is $N(0,1)$; with probability $\pi$, $Z$ is $N(\theta,1)$ with $\theta<0$. We fix $K=5000$ hypotheses, and experiment with a range of values for $\pi, \theta$. 

Our results are summarized in Table \ref{tab-fixedsignal}. 
As expected, $FDR\leq pFDR\leq mFDR$. The gain in power with the novel procedures (OMT-FDR, OMT-pFDR) is small when the mFDR of the novel procedures is only slightly above the nominal level. However, when the gain is large, the mFDR of the novel procedures can be large. The mFDR of the OMT-FDR and OMT-pFDR procedures is above 0.16 for $\theta=-1.5$, and the power gain over OMT-mFDR is more than 30\%. It is  above 0.07 for $\theta=-2, \pi=0.1$, and the gain in power is at least 4\%. It is close to the nominal level in the three other settings and the power gain is negligible. The power gain is due to the tendency of FDR and pFDR controlling policies to make very few or very many rejections with nonnegligible probability when the signal is weak or rare, and this erratic behaviour results in high $\mE(V)$ and mFDR. 
Interestingly, when the power gain is large, the FDR of the OMT-mFDR procedure is not much smaller than the nominal level. So the OMT-mFDR has lower power, but approximately the same FDR level, as OMT-FDR. The Oracle BH procedure has FDR level identical to the nominal level, as expected, and its mFDR is only slightly above the nominal level except in the weakest setting with $\pi=0.1$, where it is inflated to be 0.066.   

The last column in Table \ref{tab-fixedsignal} demonstrates clearly where OMT-FDR and OMT-pFDR differ.  In order to control the FDR, the OMT-FDR procedure either makes no rejections, or makes many rejections, when the signal is weak. As a consequence, the false discovery proportion (FDP) is either zero or much higher than the nominal level. This is perhaps an unattractive behavior of the OMT-FDR procedure.  As the signal strengthens, the probability of no rejections decreases for OMT-FDR, and its policy approaches that of OMT-pFDR. Since (arguably) pFDR is a more appropriate error measure to control than FDR for the two-group model, the more attractive OMT-pFDR policy may be preferred over OMT-FDR.

 \begin{table}[ht]
\caption{Results for $K=5000$ $z$-scores generated independently from the two-group model $(1-\pi)\times N(0,1)+\pi\times N(\theta,1)$. For each  $\theta \in \{-2.5, -2.0, -1.5\}$ and  $\pi\in\{0.1,0.3 \}$, we provide the expected number of true positives (TP=$\mE(R-V)$), FDR, pFDR, mFDR,  and probability of no rejection ($\textrm{Pr}(R=0)$), for the four procedures compared. Since $FDR = pFDR\times (1-\textrm{Pr}(R=0))$  column 5 can be  determined from columns 6 and 8. 
When the OMT-FDR policy has $\textrm{Pr}(R>0)=1$, it coincides with the OMT-pFDR policy and therefore the OMT-pFDR line is omitted. TP is bold in the settings where the power advantage of OMT-FDR and OMT-pFDR over the alternatives is non-negligible.  
}\label{tab-fixedsignal}
\centering
\begin{tabular}{|rrr|rrrrr|}
  \hline

$\pi$ & $\theta$ & Procedure & TP  & FDR  & pFDR & mFDR & $\textrm{Pr}(R=0)$ \\ \hline
 0.1& -1.5 & OMT-FDR & {\bf 29.763} & 0.050 & 0.841 & 0.843 & 0.940 \\
   && OMT-pFDR  & {\bf 12.488} & 0.045 & 0.051 & 0.824 & 0.118 \\  
  && OMT-mFDR   & 4.062 & 0.049 & 0.050 & 0.050 & 0.013\\ 
  && Oracle BH  & 6.123 & 0.050 & 0.056 & 0.066 & 0.113  \\  \hline
0.1& -2 & OMT-FDR  & {\bf 60.308} & 0.050 & 0.065 & 0.079 & 0.230 \\
   && OMT-pFDR    & {\bf 59.755} & 0.050 & 0.050 & 0.073 & 0.000 \\  
  && OMT-mFDR   & 56.403 & 0.050 & 0.050 & 0.050 & 0.000 \\ 
  && Oracle BH  & 57.277 & 0.050 & 0.050 & 0.052 & 0.000 \\  \hline
 0.1&  -2.5 & OMT-FDR  & 179.468 & 0.050 & 0.050 & 0.051 & 0.000\\
  && OMT-mFDR   & 178.992 & 0.050 & 0.050 & 0.050 & 0.000\\ 
  && Oracle BH  & 179.346 & 0.050 & 0.050 & 0.050 & 0.000  \\  \hline
0.3& -1.5 & OMT-FDR & {\bf 167.662} & 0.050 & 0.181 & 0.184 & 0.723 \\
   && OMT-pFDR  & {\bf 155.652} & 0.050 & 0.050 & 0.166 & 0.000 \\  
  && OMT-mFDR   &  117.088 & 0.050 & 0.050 & 0.050 & 0.000 \\ 
  && Oracle BH  &  118.419 & 0.050 & 0.050 & 0.051 & 0.000 \\  \hline
0.3& -2 & OMT-FDR  & 500.0330 & 0.0500 & 0.0500 & 0.0504 & 0.0000 \\
  && OMT-mFDR    & 499.3813 & 0.0500 & 0.0500 & 0.0500 & 0.0000\\ 
  && Oracle BH  & 499.7893 & 0.0500 & 0.0500 & 0.0501 & 0.0000\\  \hline
 0.3&  -2.5 & OMT-FDR  & 927.8398 & 0.0500 & 0.0500 & 0.0501 & 0.0000 \\
  && OMT-mFDR    & 927.7303 & 0.0500 & 0.0500 & 0.0500 & 0.0000\\ 
  && Oracle BH   & 927.8105 & 0.0500 & 0.0500 & 0.0501 & 0.0000  \\  \hline
\end{tabular}
\end{table}

\subsection{The dependent setting}\label{subsec-simdep}
We generate the test statistics from the two-group model, with $g(\vz \mid \vh)$ a multivariate normal distribution with mean $\mu \times \vh$ and a block diagonal covariance matrix. Within each block we experiment with a range of values for $\rho$, the symmetric correlation. 

Our results are summarized in Table \ref{tab-dep}. As correlation increases the advantage of incorporating dependence into the rule increases, and the power gain can be vary large. In the two settings where at least in half the blocks the correlation is 0.5, the power increases of OMT-FDR, OMT-pFDR, and OMT-mFDR over ind-FDR, ind-pFDR, and ind-mFDR, respectively, is at least 40\%.  As expected from Proposition \ref{prop2}, ind-mFDR maintains the nominal mFDR level of 5\%.  The procedures in the last three rows are also robust to deviations from independence. However, the misspecified policies that ignore dependence for FDR and pFDR control (ind-FDR and  ind-pFDR) have an inflated error, which is at most  6.1\% in our settings. A comparison of FDR, pFDR and mFDR policies reveals that the power gain of FDR and pFDR policies over the respective mFDR policy is large  when FDP is variable, which is manifested in the high mFDR levels (15\%--19\%). As in the set of simulations with independent test statistics, we find that the variation in FDP is greater with FDR control than with pFDR control policies.

 \begin{table}[ht]
 \centering
 \caption{Results for $K=5000$ $z$-scores generated  from the following general two-group model: 
for each $i$, $h_i\sim Bernoulli(0.3)$ is sampled  independently and  the $z$-score mean is $-1.5\times h_i$. The covariance matrix $\Sigma_h$ is a block diagonal matrix with blocks of size five, diagonal entries $1+0.01\times h_i$ and off-diagonal entries with value $\rho_{b}\in \{0,0.1,0.5\}$ for block $b\in \{1,\ldots, 1000 \}$. We provide  the FDR, pFDR, and  mFDR, as well as the expected number of true positives (TP=$\mE(R-V)$), for: the OMT procedure with Err control (OMT-Err when Err is, respectively, FDR, pFDR, and mFDR); the procedure based on the  marginal local FDR (a sub-optimal test statistic) with Err control (marg-Err); the misspecified procedure that utilizes the OMT policy for Err control under the assumptions that the $z$-scores are independent  (ind-Err);  est-mFDR;  Adaptive-BH; and BH.  TP is bold in the settings where the power  of the OMT procedures that take dependence into account is substantially larger than their marginal counterparts, which  base their decisions on the marginal locFDRs.
 }\label{tab-dep} 
 \begin{tabular}{l|rrrr|rrrr|rrrr|}
&\multicolumn{4}{|c|}{$\rho_b=0.1$}&\multicolumn{4}{|c|}{$\rho_b\in \{ 0.1, 0.5\} $}&\multicolumn{4}{|c|}{$\rho_b=0.5$}\\ 
    & {\small FDR} & {\small pFDR} & {\small mFDR} & TP & {\small FDR} & {\small pFDR} & {\small mFDR}  & TP & {\small FDR} & {\small pFDR}& {\small mFDR} & TP \\ \hline
   
 {\small OMT-FDR} & .049 & .159 & .162 & 169 & .050 & .055 & .059 & {\bf 263} &  .050 & .050 & .051 & {\bf 386} \\ 
 {\small marg-FDR} & .050 & .176 & .179 & 167 & .051 & .178 & .181 & 169 & .051 & .181 & .185 & 169 \\ 
{\small ind-FDR} &  .052 & .177 & .180 & 173 & .056 & .179 & .183 & 185 & .061 & .183 & .187 & 199 \\ 
\hline
{\small   OMT-pFDR} & .051 & .051 & .147 & 166  & .050 & .050 & .058 & {\bf 263} & .050 & .050 & .051 & {\bf 386}\\ 
{\small marg-pFDR} & .050 & .050 & .163 & 158 & .049 & .049 & .164 & 154 & .051 & .051 & .168 & 159 \\ 
{\small ind-pFDR} & .052 & .052 & .163 & 163 & .053 & .053 & .166 & 168 & .059 & .059 & .171 & 183 \\ 
\hline
{\small   OMT-mFDR} & .050 & .050 & .050 & 130 & .050 & .050 & .050 & {\bf 261} & .050 & .050 & .050 & {\bf 386} \\ 
{\small marg-mFDR} & .050 & .050 & .050 & 121 & .050 & .050 & .050 & 121 & .050 & .050 & .050 & 121 \\ 
{\small  ind-mFDR} & .050 & .050 & .050 & 120 & .050 & .050 & .050 & 121 & .050 & .050 & .050 & 121 \\ 
  \hline
{\small  est-mFDR} & .050 & .050 & .050 & 120 & .050 & .050 & .050 & 120 & .050 & .050 & .050 & 120 \\ 
 {\small  adaptive BH} & .050 & .050 & .051 & 122 & .050 & .050 & .052 & 122 & .050 & .050 & .052 & 122 \\ 
 {\small  BH} & .035 & .035 & .037 & 73 & .035 & .035 & .037 & 72 & .035 & .035 & .037 & 72 \\

 \end{tabular}
 \end{table}

\subsection{The effect of estimation of the mixture components in the two-group model}\label{subsec-simest}
In practice, the distributions $g(z|h=1)$, $g(z)$ and the mixture proportion $\pi$ are typically unknown. The estimation of the marginal density of the $z$-scores and of $\pi$ can be difficult, and there are many different approaches. We shall limit our investigation to fitting a bivariate mixture of normals using the R package \emph{mixfdr} available from CRAN \citep{Muralidharan10}. The estimation is done using the EM algorithm with a penalization via a Dirichlet prior on $(1-\pi,\pi)$. Estimation of the fraction of nulls is most conservative if the Dirichlet prior parameters are (1,0). In addition to this prior, we also examined the results with the Dirichlet prior parameters $(1-\hat{\pi}, \hat{\pi})$, where $\hat \pi$ is estimated by the method of \cite{Jin07}, recommended in \cite{Sun07}. 

Our results are summarized in Table \ref{tab-estfixedsignal}. 
As in the known distribution case, est-OMT-FDR has the most power, with est-OMT-pFDR a close second, even though it is no longer a necessary guarantee since the rejection region is computed using the estimated parameters from the data. For example, with $\pi=0.3$ the procedure est-OMT-FDR (which coincides with est-OMT-pFDR) has an FDR (which coincides with pFDR) below the nominal level, and compared with est-mFDR, it rejects few more hypotheses on average  if the non-conservative method is used for estimating the fraction of nulls, and many more hypotheses if the conservative method is used. However, the est-OMT-FDR can have an inflated FDR level when the fraction of nulls is fairly small (making the estimation problem more difficult). This problem is present to a lesser degree with est-OMT-pFDR. With $\pi=0.1$: the procedure est-OMT-FDR has an FDR level of 0.12 if the non-conservative method is used for estimating the fraction of nulls, and 0.06 if the conservative method is used; the procedure est-OMT-pFDR has a pFDR level of 0.11 if the non-conservative method is used for estimating the fraction of nulls, and 0.06 if the conservative method is used. 

 \begin{table}[ht]
\caption{Results for $K=5000$ $z$-scores generated independently from the two-group model $(1-\pi)\times N(0,1)+\pi\times N(-2,1)$. We provide the FDR, pFDR, mFDR,  and expected number of true positives (TP=$\mE(R-V)$), for the estimated OMT procedure with FDR control (est-OMT-FDR), with pFDR control (est-OMT-pFDR), with mFDR control (est-mFDR), and for adaptive BH.  The conservative estimation method uses the default prior $Dirichlet(1,0)$ for $(1-\pi, \pi)$; the non-conservative estimation method uses the estimator of \cite{Jin07}, which was recommended in \cite{Sun07} with supplementary R code. The standard error of the estimated FDR is at most 0.004. The est-OMT-FDR policy has $Pr(R>0)=1$ for every simulated dataset in the last two settings, so it coincides with the est-OMT-pFDR policy and therefore the OMT-pFDR line is omitted.   }\label{tab-estfixedsignal}

\centering
\begin{tabular}{|lll|rrrrr|}
  \hline

$\pi$ &   estimation method  & Procedure & TP  & FDR  & pFDR & mFDR & $\textrm{Pr}(R=0)$ \\ \hline
 0.1&non-conservative & est-OMT-FDR&  113.144 & 0.122 & 0.141 & 0.281 & 0.133 \\
   && est-OMT-pFDR & 103.826 & 0.108 & 0.108 & 0.253 & 0.000\\  
  && est-mFDR  & 49.769 & 0.045 & 0.045 & 0.048 & 0.001  \\ 
  && Adaptive BH  & 56.875 & 0.051 & 0.051 & 0.053 & 0.000 \\  \hline
 0.1 &  conservative & est-OMT-FDR  & 68.100 & 0.060 & 0.060 & 0.066 & 0.008 \\
   && est-OMT-pFDR & 67.740 & 0.059 & 0.059 & 0.066 & 0.000\\
  && est-mFDR  & 47.199 & 0.042 & 0.042 & 0.043 & 0.000\\ 
  && Adaptive BH  & 53.833 & 0.048 & 0.048 & 0.050 & 0.000  \\  \hline
0.3&non-conservative & est-OMT-FDR& 499.887 & 0.049 & 0.049 & 0.050& 0.000 \\
  && est-mFDR   & 491.689 & 0.048 & 0.048 & 0.048  & 0.000  \\ 
  && Adaptive BH  & 495.706 & 0.049 & 0.049 & 0.049 & 0.000 \\  \hline
 0.3 &  conservative & est-OMT-FDR  &   496.375 & 0.049 & 0.049 & 0.049& 0.000\\
  && est-mFDR  &387.820& 0.036 & 0.036 & 0.036 & 0.000\\ 
  && Adaptive BH  &452.535  & 0.043 & 0.043 & 0.043& 0.000  \\  \hline
\end{tabular}

\end{table}

\section{Gene expression data analysis}\label{sec-gene-expression}
We illustrate the utility of the novel procedures in the context of an application to gene expression studies. The goal of gene expression studies is to identify the genes that are associated with a trait of interest. For this purpose, the gene expression and the traits of  individuals are collected.

In this section, we provide our re-analysis of a meta-analysis of gene expression studies described in \cite{Shah16}, using our novel est-OMT-FDR procedure (which coincides with est-OMT-pFDR) and the competitors est-mFDR, adaptive BH, and BH. 
In \S~\ref{app-adeptus-example} we analyzed 20 additional gene-expression studies using the same procedures. Our analyses demonstrate clearly that the novel procedure tends to make the largest number of rejections. Of course, having a larger number of rejections does not guarantee having a larger number of true rejections.   We chose the example from  \cite{Shah16} since it contained  both a discovery and a validation meta-analysis study, so we could combine these to form a set of ``confirmed discoveries". A comparison of the rejections by each method with the confirmed discoveries 
suggests that est-OMT-FDR  has more  true discoveries while still maintaining a low false discovery proportion. All analyses can be reproduced using our code available from \url{https://github.com/ruheller/OMT2GroupModel}. Specifics follow. 

 \cite{Shah16} carried out  a primary meta-analysis of four studies of ulcerative colitis, and they reported the primary meta-analysis $p$-values for up-regulation, and separately, down-regulation, of genes. Using the BH procedure on the meta-analysis $p$-values,  2211 and 1775 genes with higher or lower expression, respectively, in ulcerative colitis compared with healthy controls, were detected. \cite{Shah16} also carried out a follow-up meta-analysis of four additional studies of ulcerative colitis. Finally, they carried out  a replication analysis  that showed a high concordance of the average fold-change and a significant overlap in genes with increased or decreased expression. 

We first transformed the primary meta-analysis $p$-values (based on four studies) aimed at discovering up-regulation, and separately the down-regulation,  to $z$-scores. In order to avoid unbounded values, $z$-scores with a one-sided $p$-value of zero (or very nearly zero) were sampled from $N(-6,1)$, and $z$-scores with a one-sided $p$-value of one (or very nearly one) were sampled from $N(6,1)$. 
 We note that although up-regulation and down-regulation are opposite in terms of effect sizes, the  primary meta-analysis $p$-values are directional and therefore the discoveries of interest are those that correspond to the small primary meta-analysis $p$-values for up-regulation, and separately for down-regulation. Similarly, after conversion to $z$-scores the discoveries of interest are those that correspond to small (negative) $z$-scores for  up-regulation, and separately for down regulation. We term these $z$-scores the discovery $z$-scores, since they are based only on the primary meta-analyses (i.e., excluding the follow-up meta-analyses). 

We assume the discovery $z$ scores are independently generated from a mixture of five normal densities, including  one standard normal density (corresponding to genes with no up or down regulation). 
We estimated the mixture components using the R package \emph{mixfdr} by \cite{Muralidharan10}.  The estimation carried out  uses the EM algorithm with a penalization via a Dirichlet prior with a parameter value of one for standard normal component and of zero for the remaining four mixture components.  We denote by $(\hat \mu_i, \hat \sigma_i)$  the estimated mean and standard deviation for normal component, and by  $\hat \pi_{\mu_i}$  the estimated probability the $z$-score is sampled from $N(\hat \mu_i, \ \sigma_i)$, for $i\in \{1,2,3,4\}$. For each of the two analyses we carried out (for discovering up-regulation,  and separately, down regulation), the estimated mixture density had two negative means, $(\hat \mu_3,\hat \mu_4$), and two positive means, ($\hat \mu_1,\hat \mu_2$).

In this example, the null hypothesis is a compound null, that the $z$-score has expectation at least zero.  Since $\sum_{i=1}^4\hat \pi_{\mu_i}+\hat \pi_0=1$,  the probability of the discovery $z$-scores coming from a nonnull hypothesis is estimated to be $\hat \pi = \hat \pi_{\mu_3}+\hat \pi_{\mu_4}$. 
The estimated null and alternative distributions are therefore $\hat g(z\mid h=0) = \frac{\hat \pi_0}{1-\hat \pi}\phi(z)+\frac{\hat \pi_{ \mu_1}}{1-\hat \pi}\frac1{\hat \sigma_1}\phi\left(\frac{z-\hat \mu_1}{\hat \sigma_1}\right)+\frac{\hat \pi_{ \mu_2}}{1-\hat \pi}\frac1{\hat \sigma_2}\phi\left(\frac{z-\hat \mu_2}{\hat \sigma_2}\right)$ and $\hat g(z\mid h=1) = \frac{\hat \pi_{ \mu_3}}{\hat \pi}\frac1{\hat \sigma_3}\phi\left(\frac{z-\hat \mu_3}{\hat \sigma_3}\right)+\frac{\hat \pi_{\mu_4}}{\hat \pi}\frac1{\hat \sigma_4}\phi\left(\frac{z-\hat \mu_4}{\hat \sigma_4}\right), $ respectively, where $\phi(\cdot)$ is the density of the standard normal distribution.   The  estimated marginal locFDR value for a $z$-score $z$ is therefore $T_{marg}(z) = \frac{(1-\hat \pi)\hat g(z\mid h=0)}{(1-\hat \pi)\hat g(z\mid h=0)+\hat \pi\hat g(z\mid h=1)}$.

The confirmed discoveries were identified using the discovery $z$-vector and the validation $z$-vector, combined by Fisher's combining method.  So the $p$-value for a gene with values $(zd, zv)$ for the $z$-scores in the discovery study and the validation study, respectively, is the  probability that a chi-square distribution with four degrees of freedom is larger than $-2\log(\phi(zd))-2\log(\phi(zv))$. We expect the power to be greater when basing the inference on the pooled evidence from both the primary and follow-up meta-analyses, if the genes differentially expressed in the primary meta-analyses are also differentially expressed in the  follow-up meta-analyses. Applying  a multiple testing procedure,  which is supposed to yield a negligible amount of false positives, on these Fisher combined $p$-values, we receive trustworthy discoveries which we label as the  confirmed discoveries.

In Table \ref{tab-ulcer} we see that for both up-regulation, and separately down-regulation, est-OMT-FDR has the most rejections, followed in decreasing order by est-mFDR, adaptive BH,  and lastly the BH procedure. Moreover, this order of the number of rejections is retained also when restricted to the genes in the ``set of confirmed discoveries" using the BH procedure at the 0.05 level on the Fisher combined $p$-values.  A more conservative definition of ``set of confirmed discoveries" for which no false positives are expected, by the discoveries using the Bonferroni-Holm procedure at the 0.05 level, also retains this order of the number of rejections (not shown).

 \begin{table}[ht]
 \centering
 \caption{Results for the $K=15247$ genes for up-regulation (rows 1--2) and down-regulation (rows 3--4). We provide the number of rejections by the novel procedure (column 2) and the  competitors, as well as the number  of these rejections that  are in the ``set of confirmed discoveries". }\label{tab-ulcer}
 \begin{tabular}{lllll}
   \hline
   & est-OMT-FDR & est-mFDR & adaptive BH & BH \\ 
  \hline \# up regulated & 2409 & 2305 & 2264 & 2211 \\ 
  \# up regulated in ``set of confirmed discoveries"& 2276 & 2219 & 2189 & 2145 \\ 
   \hline
 \# down regulated  & 2023 & 1897 & 1837 & 1775 \\ 
\# down regulated in ``set of confirmed discoveries"& 1815 & 1731 & 1699 & 1671 \\ 
 \end{tabular}
 \end{table}


\section{Discussion}

In this paper, we provide the first practical approach to the problem of maximizing an objective which is linear in the decision functions, subject to FDR or pFDR control in the general two-group model. With the generic form of our formulation for finding the OMT policies in \eqref{prob-genericform}, it is also possible to consider other error criteria, e.g.,
 FWER ($Pr(V>0)$), or false discovery exceedance ($Pr(FDP>\gamma)$). As with FDR control, the optimal solution will be to threshold the locFDR at a value that depends on  the $K$ realized locFDR statistics. The error measures $\mE(V)$ and mFDR result in a much simpler solution (see derivation in \S~\ref{app-mfdr-rule} for mFDR control), where the threshold for rejection depends only on the mixture distribution. It is also possible to consider novel criteria, such as the probability of false discovery exceedance given that at least one rejections occurred, $\textrm{Pr}\left(\frac{V}{R}>\gamma \mid R>0 \right)$. Moreover, the formulation can be extended in a straightforward manner to control more than one error rate. For example, seek the OMT policy which controls for the FDR as well as for $\mE(V)$, thus potentially creating a powerful policy with meaningful control over the false discovery proportion in expectation without allowing for an unattractive policy which tends to reject many or very few hypotheses.

We provide efficient algorithms for computing the optimal policy for independent test statistics, as well as if the test statistics are equi-correlated or have a block dependence covariance structure. 
We demonstrate the large potential power gain from incorporating the dependence into the OMT policies. We expect the OMT policies  to be useful in genomic applications where the dependence is known. For example, in Genome-Wide association studies (GWAS), the covariance matrix is a known banded matrix, due to linkage disequilibrium. We plan to  provide efficient computational tools for the general two-group model with this type of local dependence. 

Our general two-group model assumes that the hypotheses states are independently sampled. \cite{Sun09} considered the setting where the underlying latent indicator variable of being null follows a homogeneous irreducible hidden markov (HMM) chain. In their setting, the test statistics conditional on  the hypotheses states are independent. Deriving solutions with an HMM structure in our framework is also an interesting direction for future work.

In the two-group model, the potential gain in power from applying optimal policies with FDR or pFDR control rather than mFDR control is maintained when the parameters are estimated in our numerical experiments, but care has to be taken in proper estimation of the mixture parameters in order to avoid an unacceptable inflation of the FDR or pFDR level. In particular, it appears that the estimation of the fraction of nulls has to be conservative when the actual fraction is fairly small. Further research into estimation methods tailored towards estimated optimal policies for FDR and pFDR control  is needed. {In our real data examples, the most discoveries (as well as the most validated discoveries, confirmed using a validation set) were made with the estimated optimal marginal policy for FDR control. In these examples, the test statistics are not independent but the dependence is unknown and expected to be limited in range and magnitude. Further research is needed to understand when local dependence can be ignored, as well as towards the development of  robust estimation methods for the mixture components.  }

\bibliographystyle{apalike}

\appendix

\newcommand{\vl}{\ensuremath{\vec l}}
\newcommand{\vE}{\ensuremath{\vec E}}
\section{Proofs and additional mathematical details}\label{app-proofs}

\subsection*{Proof of Proposion \ref{prop1}}

Item 1 follows straightforwardly from the explanation in the paragraph leading to the proposition. 

Item 2 follows from the fact that OMT-mFDR is a single step procedure, yet OMT-pFDR is by construction the  step-down procedure described in \S~\ref{sec-algorithm}. Put another way, the necessary conditions for the OMT-mFDR policy lead to the single step procedure, and the OMT-pFDR policy does not satisfy these necessary conditions. For example, for $K=2$, 
let $\vE(\vz)$ and $\vD(\vz)$ be the OMT-mFDR and OMT-pFDR policies, respectively. 
Then $\{(T(z_1),T(z_2)): E_{1}(\vz) =1\} = \{(T(z_1),T(z_2)): T(z_1)\leq c\}$ for a constant $c$ which guarantees $mFDR(\vE) = \alpha$, but 
$\{(T(z_1),T(z_2)): D_1(\vz)=1\} = \left \lbrace (T(z_1),T(z_2)): T(z_1)\leq  \frac {1+\alpha}{1+\mu} \  \textrm{or} \ T(z_1)+T(z_2)\leq   \frac {2}{1+\mu/2}\right \rbrace$
 for a constant $\mu$ which guarantees $pFDR(\vD)=\alpha$. Clearly, the symmetric difference between the sets $\{(T(z_1),T(z_2)): E_{1}(\vz) =1\}$ and $\{(T(z_1),T(z_2)): D_1(\vz)=1\}$ has positive Lebesgue measure. 

For item 3, suppose by contradiction that $mFDR\leq pFDR$ for the OMT-pFDR policy. The OMT-pFDR policy is necessarily at least as powerful as the OMT-mFDR policy since the OMT-mFDR policy  controls the pFDR. So the OMT-pFDR policy is optimal for mFDR control if it satisfies $mFDR\leq pFDR$. But to achieve optimal mFDR control, a policy  has to satisfy 
 necessary conditions  which lead to the single step procedure. 
 This contradicts the fact that the OMT-pFDR policy is necessarily not a single step procedure, as shown for item 2 above.

Item 4 follows by the same reasoning as that of item 1. The OMT-FDR policy is necessarily at least as powerful as the OMT-pFDR policy since the OMT-pFDR policy  controls the FDR (which is bounded above by the pFDR). Therefore, if the OMT-FDR policy  controls the pFDR (since the probability of no rejections is zero), this must be the OMT-pFDR policy as well. Indeed, it is easy to see  that the step-down procedures for optimal pFDR and optimal FDR control  in  \S~\ref{sec-algorithm} coincide when the hypothesis with minimal locFDR is rejected with probability one.

\subsection*{Proof of Theorem \ref{thm-main}} 


Given a candidate solution $\vec D$, we prove the lemma by constructing an alternative solution $\vE$ that complies with the condition and has no lower objective and no higher constraint than $\vec D$.

For every pair of indexes $1\leq i<j\leq K$, define:
$$ A_{ij} =\left\{\vec z :  T_i(\vz)< T_j(\vz),\;  D_i(\vec z) =0, D_j(\vec z)=1 \right\}.$$

We will now examine the solution $\vE$ which is equal to $\vec D$ everywhere, except on the set $A_{ij}$, where it switches the value of coordinates $i,j$:
$$
E_k(\vz) = \left\{ \begin{array}{ll} D_k(\vz) & \mbox{if } \vz \notin A_{ij} \mbox{ or } k\notin \{i,j\}\\
1-D_k(\vz) & \mbox{if } \vz \in A_{ij} \mbox{ and } k\in \{i,j\} 
\end{array} \right. .
$$

We now show the following:
\begin{enumerate}
    \item For the integrated power in Eq. (\ref{prob}), $\Pi(\vE) \geq \Pi (\vec D)$.\label{lemma1-item1} 
    \item For the Error constraint in Eq. (\ref{const}), $Err(\vE) \leq  Err(\vec D)$.\label{lemma1-item2}
\end{enumerate}
Therefore $\vE$ is an improved solution compared to $\vec D$. This can be done for all $i,j$ pairs repeatedly until $\mP(A_{ij}) = 0 \forall i,j$, and we end up with $\vE$ which has the desired monotonicity property and is superior to $\vec D$. 
Since $\vec D^*$ the optimal solution cannot be improved, it must have this monotonicity property. 

It remains to prove properties \ref{lemma1-item1} and  \ref{lemma1-item2} above. For the power, we write the expression in Eq. (\ref{prob}) for $\vE$ and $\vD$ and subtract them:
\begin{eqnarray*}
\Pi(\vE) - \Pi(\vD) &=& \int_{\mR^K}  \sum_{l=1}^K E_l(\vz) (1-T_l(\vz)) \PP(\vz) d\vz - \int_{\mR^K}  \sum_{l=1}^K D_l(\vz) (1-T(z_l)) \PP(\vz) d\vz \\ &= &\int_{A_{ij}}  (T_j(\vz)-T_i(\vz))  \PP(\vz) d\vz \geq 0,
\end{eqnarray*}
where the second equality uses the definition of $\vE$
, and the  inequality follows since $T_j(\vz)>T_i(\vz)$ for  $\vz \in A_{ij}$.

The same idea applies to the FDR constraint:
\begin{eqnarray*}
FDR(\vD) &-& FDR(\vE) =  
\int_{\mR^K}  \sum_{l=1}^K \frac{E_l(\vz)}{\vone'\vE(\vz)} T_l(\vz) \PP(\vz) d\vz - \int_{\mR^K}  \sum_{l=1}^K \frac{D_l(\vz)}{\vone'\vD(\vz)} T_l(\vz) \PP(\vz) d\vz 
\\ &= &\int_{\mR^K}  \sum_{l=1}^K \frac{E_l(\vz)}{\vone'\vE(\vz)} T_l(\vz) \PP(\vz) d\vz - \int_{\mR^K}  \sum_{l=1}^K \frac{D_l(\vz)}{\vone'\vE(\vz)} T_l(\vz) \PP(\vz) d\vz 
\\ &= &\int_{A_{ij}}  \frac{1}{\vone'\vE(\vz)}(T_j(\vz)-T_i(\vz))  \PP(\vz) d\vz
\geq 0,
\end{eqnarray*}
where the second and third equalities follow since the difference between the two ratios is nonzero only in the numerator and  only on $A_{ij}$.

It remains to show that $pFDR(\vD) - pFDR(\vE) \geq 0$. This clearly follows since for $\vz \in A_{ij}$,  $\vec 1 ^t \vD(\vz) = \vec 1 ^t \vE(\vz)$, so 
\begin{eqnarray}
&& \int_{\mR^K} \II\{  \vec 1 ^t \vE(\vz)>0 \}  \PP(\vz) d\vz -\int_{\mR^K} \II\{  \vec 1 ^t \vD(\vz)>0 \} \PP(\vz) d\vz\nonumber \\ && =\int_{A_{ij}}  (\II\{  \vec 1 ^t \vE(\vz)>0 \}- \II\{  \vec 1 ^t \vD(\vz)>0 \}) \PP(\vz) d\vz=\int_{A_{ij}} 0\PP(\vz) d\vz = 0.\nonumber
\end{eqnarray}

So the denominators in $pFDR(\vD)$ and $pFDR(\vE)$ are the same, and hence  
$$pFDR(\vD) - pFDR(\vE) =  \frac{FDR(\vD)-FDR(\vE)}{\int_{\mR^K}   \II\{  \vec 1 ^t \vE(\vz)>0 \} \PP(\vz) d\vz}\geq 0.
$$

\subsection*{Derivation of Euler-Lagrange conditions for Problem (\ref{prob-lin})}

Our optimization problem is:  
\begin{equation*}
\begin{aligned}
& {\text{max}}
& & \int_{\mR^K} \sum_k a_k (\vec z) \tD_k(\vec z) \mP(\vz) d \vec z\\
& \text{s.t.}
& &  \int_{\mR^K} \sum_k b_{k}(\vec z)\tD_k(\vec z) \mP(\vz)  d \vec z \leq \alpha \quad\quad  \\
&
& & 0 \leq \tD_K(\vec z) \leq \dots \leq  \tD_j(\vec z) \leq \tD_i(\vec z) \leq \dots \leq \tD_1(\vec z) \leq 1 \quad\quad \forall \vz\in \mR^K.
\end{aligned}
\end{equation*}
We eliminate the inequality constraints, by introducing non-negative auxiliary variables, and then square those variables to also eliminate non-negativity constraints: \begin{equation}
\label{eq:E(V)}
\begin{aligned}
& {\text{max}}
& & \int_{\mR^K} \sum_k a_k (\vec z) \tD_k(\vec z) \mP(\vz) d\vec z\\
& \text{s.t.}
& &  \int_{\mR^K} \sum_k b_{k}(\vec z)\tD_k(\vec z) \mP(\vz) d\vec z + E^2 = \alpha  \\
&
& & \tD_K(\vec z)=e^2_{K+1}(\vz) \quad \quad  \quad \quad \quad  \quad \quad \forall \vec z\in \mR^K\\
&
& &  \tD_k(\vec z) - \tD_{k+1}(\vec z) = e^2_{k+1}(\vz) \quad \quad \;\; \forall 0<k<K, \; \vec z\in \mR^K\\
&
& &  1- \tD_1(\vec z) = e^2_1(\vz) \quad \quad \quad \quad \quad \quad \forall \vec z\in \mR^K\\
\end{aligned}
\end{equation}

The Euler-Lagrange (EL) necessary conditions for a solution to this optimization problem may be obtained through calculus of variations \citep{korn2000mathematical}. Let $y_1(x),y_2(x), \dots ,y_n(x): \mathbb{R} \rightarrow \mathbb{R}$ be a set of $n$ functions and
\begin{equation}
\label{cov_objective}
I=\int_{x_0}^{x_F} F(y_1(x),y_2(x), \dots ,y_n(x); y_1'(x),y_2'(x), \dots, y_n'(x) ;x)dx
\end{equation}
be a definite integral over fixed boundaries $x_0,x_F$. Every set of $y_1(x),y_2(x), \dots ,y_n(x)$ which maximize or minimize (\ref{cov_objective}) must satisfy a set of $n$ equations
\begin{equation}
\label{cov_conditions}
\frac{d}{dx}\left(\frac{\partial F}{\partial y'_i}\right)-\frac{\partial F}{\partial y_i}=0 \quad\quad i=1,\dots,n.
\end{equation}

In addition, let
\begin{equation}
\label{cov_constraints1}
\varphi_{j_1} (y_1(x),y_2(x), \dots ,y_n(x);x)=0  \quad\quad j_1=1,\dots,m_1 < n,
\end{equation}
be a set of $m_1<n$ point-wise equality constraints on $y_1(x),y_2(x), \dots ,y_n(x)$ and
\begin{equation}
\label{cov_constraints2}
\int_{x_0}^{x_F} \Psi_{j_2} (y_1(x),y_2(x), \dots ,y_n(x);y'_1(x),y'_2(x), \dots ,y'_n(x);x)=C_{j_2}  \quad\quad j_2=1,\dots,m_2,
\end{equation}
be a set of $m_2$ integral equality constraints on $y_1(x),y_2(x), \dots ,y_n(x)$. Then, every set of $n$ functions $y_1(x),y_2(x), \dots ,y_n(x)$ which maximize (\ref{cov_objective}), subject to the constraints (\ref{cov_constraints1}, \ref{cov_constraints2}) must satisfy the EL equations,
\begin{equation}
\label{cov_conditions2}
\frac{d}{dx}\left(\frac{\partial \Phi}{\partial y'_i}\right)-\frac{\partial \Phi}{\partial y_i}=0 \quad\quad i=1,\dots,n,
\end{equation}
where
\begin{equation}
\label{bla}
\Phi=F-\sum_{j_1=1}^{m_1}\lambda_{j_1}(x)\varphi_{j_1}-\sum_{j_2=1}^{m_2}\mu_{j_2} \Psi_{j_2}.
\end{equation}
The
unknown functions $\lambda_{j_1}(x)$ and constants $\mu_{j_2}$ are called the Lagrange multipliers. The differential equations in (\ref{cov_conditions2}) are necessary conditions for a maximum, provided that all the quantities on the left hand side of (\ref{cov_conditions2}) exist and are continuous.

Hence,  the set of $y_1(x),y_2(x), \dots ,y_n(x)$ which maximize (\ref{cov_objective}) subject to the constraints (\ref{cov_constraints1},\ref{cov_constraints2}), is to be determined, together with unknown Lagrange multipliers,  from  (\ref{cov_constraints1},\ref{cov_constraints2},\ref{cov_conditions2}).

This derivation may also be extended to a higher dimensional case, $x, y_1(x),y_2(x), \dots ,y_n(x) \in \mathbb{R}^d$, as appears in \cite{korn2000mathematical}. In this case the EL equations are
\begin{equation}
\label{cov_conditions3}
\sum_{k=1}^d \frac{\partial}{\partial x_k}\left(\frac{\partial \Phi}{\partial y_{i,k}}\right)-\frac{\partial \Phi}{\partial y_i}=0 \quad\quad i=1,\dots,n,
\end{equation}
where $y_{i,k}\triangleq\frac{\partial y_i}{\partial x_k}$ and $\Phi$ follows the same definition as in (\ref{bla}), with
\begin{eqnarray*}
& \int \Psi_{j_2} (y_1(x),y_2(x), \dots& ,y_n(x);y_{1,1}(x),y_{1,2}(x), \dots ,y_{1,d}(x), \dots,  y_{n,1}(x),y_{n,2}(x), \dots ,y_{n,d}(x) ;x)=\\
&&C_{j_2}  \quad\quad j_2=1,\dots,m_2 .
\end{eqnarray*}

Therefore, the Lagrangian $\Phi$ for our optimization problem (\ref{eq:E(V)}) is
\begin{align} \label{lagrangian}
\Phi =& \sum_k a_k (\vec z) \tD_k(\vec z) - \mu \left(\sum_k b_{k}(\vec z)\tD_k(\vec z)+E^2\right) -\\\nonumber
&\lambda_{K+1}(\vec z) \left(e^2_{K+1}(\vec z)-\tD_K(\vec z)\right)- \sum_{k=2}^{K} \lambda_{k}(\vec z) \left(e^2_k(\vec z) + \tD_{k}(\vec z) - \tD_{k-1}(\vec z)\right)- \lambda_{1}(\vec z) \left( \tD_1(\vec z)+e^2_1(\vec z)-1\right).
\end{align}

The necessary conditions for the minimizers of (\ref{eq:E(V)}) are that the original constraints are met with equality, and additionally
\begin{enumerate}
\item $\frac{\partial \Phi}{\partial  \tD_k(\vz)}=\mP(\vz) \left(a_k(\vec z)- \mu b_{k}(\vec z)\right)+\lambda_{k+1}(\vec z)-\lambda_{k}(\vec z)=0 \quad \quad \forall 1\leq k \leq K,; \vz\in \mR^K $   \label{derivative_1}
\item $\frac{\partial \Phi}{\partial  e_{k}(\vz)}=2e_{k}(\vec z)\lambda_{k}(\vec z) =0 \quad \quad \forall 1\leq k \leq K+1,\;\vz\in \mR^K$ 
\label{derivative_2}
\item $\frac{\partial \Phi}{\partial  E}=2\mu E= 0\quad \quad$
\label{derivative_3}
\end{enumerate}

It is interesting to notice that these condition are exactly the KKT conditions for the discrete optimization case, where $\vz$ is over a finite grid. Specifically, the first condition corresponds to the derivatives of the Lagrangian, while conditions (\ref{derivative_2}), (\ref{derivative_3}), are equivalent to the complementary slackness property.

Note also the multiplication by $\mP(\vz)$ in the first condition, which plays no role (since the scale of $\lambda_k(\vz)$ is arbitrary) and is eliminated in the main text for simplicity.  

\subsection*{Proof of Lemma \ref{lem2}}
Assume that for some $\vec z\in \mR^K$ and index $j$ we have that $0<\tD_j(\vec z)<1$. Then it is easy to see that out of the $K+1$ constraints implied by conditions \eqref{KKT-CS2}--\eqref{KKT-CS4}, at least two will require $\lambda_i = 0$ to hold: for example, if $0<\tD_1(\vec z) <1$ and $\tD_2(\vec z)=\ldots=\tD_K(\vec z)=0$, we will have that $\lambda_1(\vec z)=\lambda_2(\vec z) = 0$ to maintain complementary slackness.

Assume wlog that $\lambda_l(\vec z)=\lambda_j(\vec z)=0$ for some $l<j$. Now we can sum the equations between $l$ and $j-1$ in the stationarity condition (\ref{KKT-stat}):
\begin{eqnarray*}
&&\sum_{i=l}^{j-1} \left[ a_i(\vec z)- \mu b_{i}(\vec z)-\lambda_{i}(\vec z)+\lambda_{i+1}(\vec z) \right]= \\ &&\;\;\;\;\;\; \sum_{i=l}^{j-1} \left[ a_i(\vec z)- \mu b_{i}(\vz)\right] = 0,
\end{eqnarray*}
where all the $\lambda$ terms have cancelled out due to the telescopic nature of the sum, and $\lambda_l=\lambda_j=0$.

Hence we have concluded that having any non-binary value in the optimal solution $\tD^*(\vec z)$ implies
$$
 \sum_{i=l}^{j-1} \left\lbrace a_i(\vec z)- \mu b_{i}(\vec z)\right\rbrace =  0,
$$
which has probability zero since, by our assumption for the two-group model, $\sum_{i=l}^{j-1} \left\lbrace a_i(\vec Z)- \mu b_{i}(\vec Z)\right\rbrace$ is a continuous random variable. 

\subsection*{Derivation of dual to Problem  (\ref{prob-lin}) and proof of Lemma \ref{lem3}}

The result in Lemma \ref{lem3} relies on explicit derivation of the dual to the  infinite linear program (\ref{prob-lin}) (see \cite{Anderson87} for details on derivation of dual to infinite linear programs):
\begin{eqnarray}
\min_{\mu,\lambda} &&\alpha \mu  + \int_{\mR^K} \lambda_1(\vec z) d \vec z  \label{eq-dual}\\
\mbox{s.t.}&& a_k(\vec z) - \mu b_{k}(\vec z) + \lambda_{k+1}(\vec z) - \lambda_k(\vec z) \leq 0,\;\forall k, \vec z \nonumber \\
&&\lambda_k(\vec z) \geq 0,\;\forall k,\vec z \;;\;\;\mu \geq 0. \nonumber
\end{eqnarray}

\noindent{\bf Proof of Lemma \ref{lem3}:} Feasibility of dual solution holds by construction: $\mu, \lambda$ are non-negative Largange multipliers by definition, and the EL conditions require that
$$a_i(\vec z) - \mu^* b_{i}(\vec z)-\lambda^*_{i}(\vec z)+\lambda^*_{i+1}(\vec z) = 0\;,\forall i,\vec z.$$

To calculate the dual objective, we explicitly derive the value of $\lambda^*_1(\vec z)$ as a function of the other variables.
 If $\tD^*_K(\vec z)=1$, then $\lambda^*_{K+1}(\vec z) =0$ and it is easy to see from  \eqref{KKT-stat}--\eqref{KKT-CS4} that $\lambda^*_1(\vec z)$ is equal to
$$
\sum_{i=1}^K \left(a_i(\vec z) - \mu^* b_{i}(\vec z)\right).
$$
Similarly, if $\tD^*_{j-1}(\vec z)-\tD^*_j(\vec z)=1$ for $j\in \{2,\ldots, K-1\}$, then $\lambda^*_j(\vec z) = 0$ and $\lambda^*_1(\vec z)$ is equal to
$$
\sum_{i=1}^{j-1}\left( a_i(\vec z) - \mu^* b_{i}(\vec z)\right).
$$
It thus follows that $$\lambda^*_1(\vec z) = \sum_{j=1}^K \tD^*_j(\vec z)\left(a_j(\vec z) - \mu^* b_{j}(\vec z)\right). $$
Therefore, $$\int_{\mR^K} \lambda^*_1(\vec z) d\vec z = \int_{\mR^K} \left(\sum_{j=1}^K a_j(\vec z) \tD^*_j(\vec z)\right)d \vec z -  \mu^* \int_{\mR^K} \left(\sum_{j=1}^K b_{j}(\vec z) \tD^*_j(\vec z)\right). $$
Therefore the dual objective is equal to the primal objective:
\begin{eqnarray*}
&& \sum_{L=0}^{K-1} \mu^* \alpha + \int_{\mR^K} \left(\sum_{j=1}^K a_j(\vec z) \tD^*_j(\vec z)\right)d\vec z -  \mu^* \int_{\mR^K} \left(\sum_{j=1}^K b_{j}(\vec z) \tD^*_j(\vec z)\right)d \vec z \nonumber\\
&& = \mu^* \left\lbrace\alpha - \int_{\mR^K} \left(\sum_{j=1}^K b_{j}(\vec z) \tD^*_j(\vec z)\right)d\vec z\right\rbrace + \int_{\mR^K} \left(\sum_{j=1}^K a_j(\vec z) \tD^*_z(\vec z)\right)d\vec z \nonumber\\
&& =  \int_{\mR^K} \left(\sum_{j=1}^K a_j(\vec z) \tD^*_j(\vec z)\right)d\vec z,\nonumber
\end{eqnarray*}
where we have used the complementary slackness condition for the $\mu^*$ in the last equality.

\section{Derivation of the expression for $FDR(\vec D)$ in \eqref{const-lin-FDR}}\label{app-derFDR}
By Theorem  \ref{thm-main}, on Q,  ${\vone}^t \vD(\vz)=k$ if and only if $\tD_1=\ldots \tD_k=1$ and $\tD_{k+1}=\ldots=\tD_K=0$, i.e., if and only if $\tD_k-\tD_{k+1}=1$.  Therefore: 
\begin{eqnarray}
 \sum_{i=1}^K\frac{\tD_i(\vz)} {{\vone}^t \vD(\vz)}T_{(i)}(\vz)  &=& \sum_{k=1}^{K-1} \sum_{i=1}^k \frac 1kT_{(i)}(\vz)\lbrace \tD_k(\vz)-\tD_{k+1}(\vz) \rbrace +\tD_K(\vz) \sum_{i=1}^K \frac 1KT_{(i)}(\vz) \nonumber\\ &=& \sum_{k=1}^{K-1} \bar{T}_k(\vz)\lbrace \tD_k(\vz)-\tD_{k+1}(\vz) \rbrace +\tD_K(\vz) \bar{T}_{(K)}(\vz) \nonumber \\ &=&
 \sum_{k=1}^{K} \bar{T}_k(\vz) \tD_k(\vz)-\sum_{k=2}^{K} \bar{T}_{k-1}(\vz) \tD_k(\vz)=\tD_1(\vz)T_{(1)}(\vz)+\sum_{k=2}^K\tD_k(\vz)\lbrace \bar{T}_k(\vz)-\bar{T}_{k-1}(\vz) \rbrace\nonumber \\ &=&
 \tD_1(\vz)T_{(1)}(\vz)+\sum_{k=2}^K\tD_k(\vz)\frac 1k\lbrace T_{(k)}(\vz)-\bar{T}_{k-1}(\vz) \rbrace\label{eq-app-B}
\end{eqnarray}

For the general two-group model, from \eqref{const-FDR} it thus follows that
$$  \int_{\mR^K}  \sum_{i=1}^K\frac{D_i(\vz)} {{\vone}^t \vD(\vz)}T_{i}(\vz)  \PP(\vz)    d\vz =  \int_{\mR^K}  \tD_{(1)}(\vz)T_{(1)}(\vz)+\sum_{k=2}^K\tD_k(\vz)\frac 1k\lbrace T_{(k)}(\vz)-\bar{T}_{k-1}(\vz)  \PP(\vz)    d\vz ,$$
where  the second equality follows from expression \eqref{eq-app-B} derived for $\tD$ that is  weakly monotone in the locFDRs.

\section{An alternative proof of the rejection policy for OMT with mFDR control}\label{app-mfdr-rule}
We shall show that the solution to the optimization problem of finding the optimal decision rule with the  expected number of true rejections as the objective and the  mFDR at  most level $\alpha$ as the constraint, coincides with the rule of \cite{Xie11} for the two-group model. 

The constraint $mFDR\leq \alpha$  is  equivalent to  $\mE(V(\vec D))-\mE(R(\vec D))\alpha\leq 0$, where 
\begin{eqnarray}
&& \mE \lbrace V(\vec D)-\alpha R(\vec D) \rbrace =(1-\alpha)\mE \lbrace V(\vec D)\rbrace-\alpha\mE \lbrace R(\vec D) - V(\vec D) \rbrace \nonumber\\
&&= (1-\alpha) \int_{\mR^K} \sum_{i=1}^K D_i(\vec z) T_i(\vz) \PP(\vz) d\vz\nonumber-\alpha \int_{\mR^K} \sum_{i=1}^K D_i(\vec z) \lbrace 1- T_i(\vz)\rbrace  \PP(\vz) d\vz\nonumber \\
&&= \int_{\mR^K} \sum_{i=1}^K D_i(\vec z)\lbrace (1-\alpha)-(1-T_i(\vz))\rbrace  \PP(\vz)  d\vz\nonumber.
\end{eqnarray}

Therefore, the linear program for maximizing the objective subject to mFDR control is \eqref{prob-genericform}
where  $b_{i}(\vz) = (1-\alpha)-(1-T_i(\vz))=T_i(\vz)-\alpha$ and $a_i(\vz) =1-T_i(\vz) $. 


 As in the FDR proof, the EL necessary optimality conditions are: 
\begin{eqnarray}
&&a_i(\vz) - \mu b_{i}(\vz)-\lambda_{i1}(\vz)+\lambda_{i2}(\vz) = 0,\;\;\forall \vz \in \mR^K, i=1,\ldots,K. \label{KKT-stat-mfdr}\\
&& \mu\left\lbrace 
\int_{\mR^K} \PP(\vz) \sum_{i=1}^K D_i(\vec z)b_i(z_i) d\vz
\right\rbrace=0,\label{KKT-CS1-mfdr} \\
&&\lambda_{i1}(\vz)D_i(\vz) = 0 \ \forall \vz\in \mR^K, i=1,\ldots,K.  \label{KKT-CS2-mfdr}\\
&&\lambda_{i2}(\vz) (D_i(\vz)-1 ) = 0\;,\;\forall \vz \in \mR^K, i=1,\ldots,K\label{KKT-CS4-mfdr},
\end{eqnarray}
where $\mu,\lambda_{ij}(\vz),\; i=1,\ldots,K, j=1,2,\;\vz\in \mR^K$ are non-negative Lagrange multiplies.
The solution that satisfies \eqref{KKT-stat-mfdr},\eqref{KKT-CS2-mfdr}, and \eqref{KKT-CS4-mfdr} is guaranteed to be an integer solution, 
since if $0<D_i(\vz)<1$ it follows that $\lambda_{i1}(\vz)=\lambda_{i2}(\vz)=0$ and therefore that $a_i(\vz)-\mu b_i(\vz) = 0$. Moreover, following steps similar to  the ones in the FDR proof of Lemma \ref{lem3}, it can be shown that conditions \eqref{KKT-stat-mfdr}-\eqref{KKT-CS4-mfdr} together with primal feasibility are sufficient.  

Clearly, given $\mu>0$, almost surely the rejection policy that satisfies \eqref{KKT-stat-mfdr},\eqref{KKT-CS2-mfdr}, and \eqref{KKT-CS4-mfdr} is  
$$D^\mu_i(\vec z) = \II\left \lbrace a_i(\vz)-\mu b_i(\vz)>0 \right\rbrace= \II\left \lbrace T_i(\vz)< \frac{1+\mu\alpha}{1+\mu} \right\rbrace.$$
Therefore, all that remains is to find $\mu$   that satisfies $\int_{\mR^K} \PP(\vz) \sum_{i=1}^K D^{\mu}_i(\vec z)b_i(\vz) d\vz=0$, i.e., 
$ \mE(V(\vec D^{\mu})) -\mE(R(\vec D^{\mu}))\alpha= 0.$




\section{Gene expression data analysis in the ADEPTUS database}\label{app-adeptus-example}

A recently compiled high-quality curated database spanning more than 38000 gene expression profiles and more than 100 phenotypes (mainly a disease or a control label) is the ADEPTUS database \cite{Amar18}. We shall use a subset of the gene expression studies in the database to illustrate our suggested methods. 

Our starting point is the one-sided $p$-values as computed in the ADEPTUS database for the 149 gene expression studies (so down-regulation will tend to have a  $p$-value close to zero, and up-regulation will tend to have  a $p$-value close to one). We estimated the fraction of differentially expressed genes using Storey's plug-in method with parameter $\lambda = 0.05$, as suggested in \cite{Blanchard09}. We selected the studies which had an estimated proportion of signal of at least 15\%. This crude selection of studies is convenient, since we expect for the vast majority of the selected studies that the procedures est-OMT-FDR and est-OMT-pFDR coincide, and therefore we only used est-OMT-FDR in our re-analysis of these datasets. Moreover, we wanted to avoid having a fraction of nulls that is too small for proper evaluation of the mixture parameters. In addition,  we restricted our attention to studies that had well behaved $p$-values, i.e. that the $z$-score corresponding to the $p$-values where all between -10 and 10. We were left with 20 studies to analyze, and each study had $K= 10081$ genes. 

We assume the $z$ scores are independently generated from a mixture of three normal densities, for up-regulation (positive expected value for the $z$-score), down regulation (negative expected value for the $z$-score), and for the null component which is assumed to have  a standard normal distribution.  As in \S~\ref{sec-gene-expression}, we estimated the mixture components using the R package \emph{mixfdr} by \cite{Muralidharan10}. 

Figure \ref{fig-adeptus} shows that est-OMT-FDR tends to make the most rejections, followed by adaptive BH, then est-mFDR, and finally BH. Thus, at least in terms of number of discoveries, the novel procedure with estimated mixture components by mixFDR has an advantage over the others.
\begin{figure}
\includegraphics{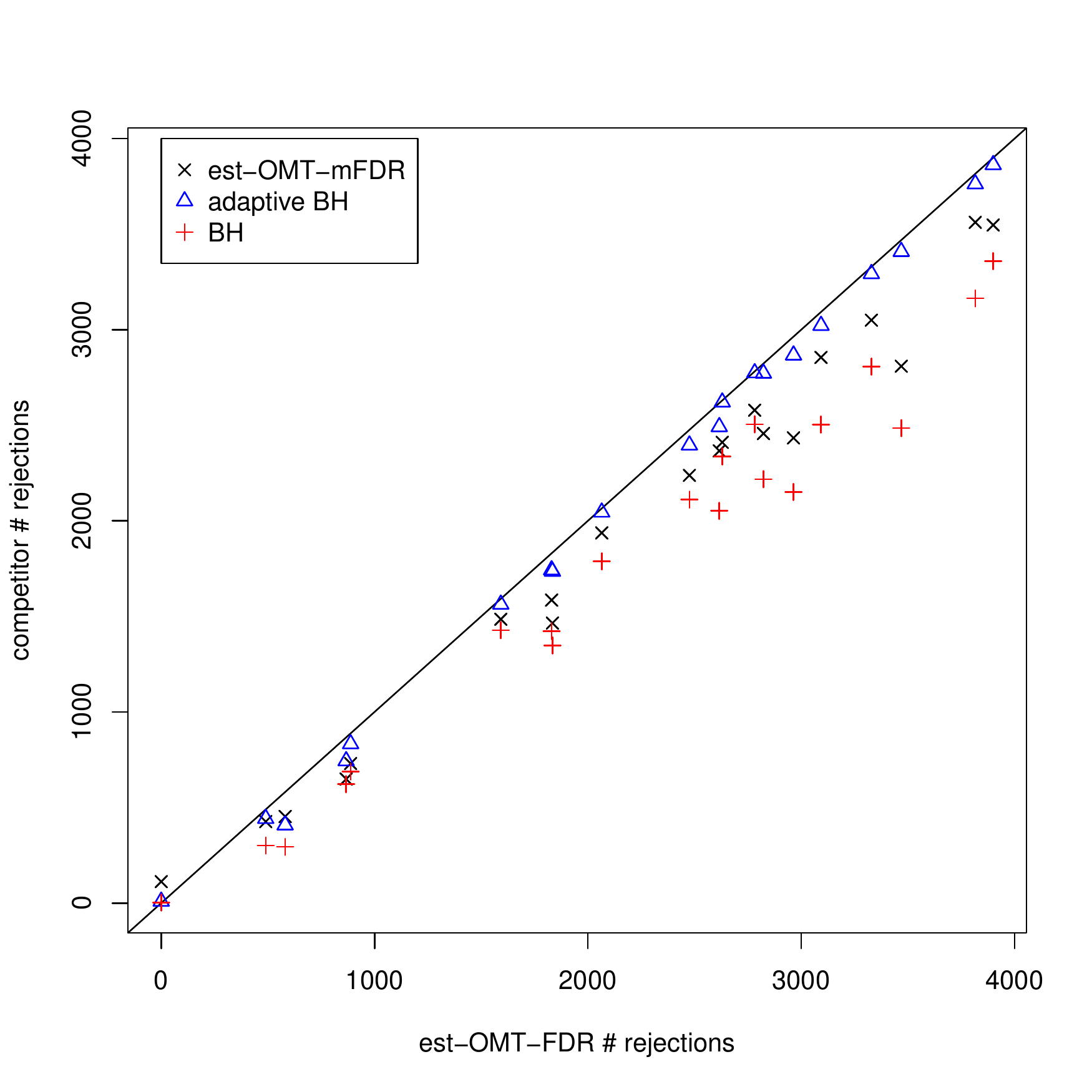}
\caption{The number of discoveries of competitors versus est-OMT-FDR. Diagonal line in solid black. }\label{fig-adeptus}. 
\end{figure}

\section{ Calculating locFDR under equi-correlation}\label{app-equicor}

As we have seen, under any type of dependence, the problem of controlling FDR, pFDR, mFDR in the two-group model boils down to calculating the localFDRs $T_i$ in a manner which takes in the depdendence structure: 

\begin{eqnarray} T_i(\vZ) &=& \mP (h_i = 0 | \vZ) = \frac{\mP(\vZ,h_i=0)}{\mP(\vZ)} = 
 \frac{\sum_{\vh:h_i=0} \mP(\vh) \mP(\vZ|\vh)}{\sum_{\vh} \mP(\vh) \mP(\vZ|\vh)}. \label{Eq:Ti}
\end{eqnarray}

The sum in the denominator contains $2^K$ terms, and the one in the numerator a subset of $2^{K-1}$ of them. For large $K$ this is not practical in the general case. 

As shown in the main text (\S \#), when the test statistics are independent, we have a simple calculation that uses only $Z_i$ to calculate $T_i$. In the main text we also discuss the case where the dependence is in relatively small blocks, in which case the complexity becomes exponential in block size rather than in $K,$ and linear in the number of blocks. 

Here we discuss a case where all entries in $\vZ$ are dependent, but the calculations can still be done efficiently with careful analysis. 

Consider the equi-correlated case  where 
$$ \vZ | \vh \sim \left(\delta \vh, \Sigma\right),$$
where $\Sigma$ and its inverse are both equi-correlation matrices: $$ \Sigma_{ij} = \left\{ \begin{array}{cc} 
\sigma^2 & \mbox{if } i=j\\
\rho \sigma^2 & \mbox{otherwise}
\end{array} \right. \;\;,\;\;\;\Sigma^{-1}_{ij} =  \left\{ \begin{array}{cc} 
a:=\frac{-1-\rho(K-2)}{\sigma^2 \left(\rho(K-1)+1\right) (\rho-1)} & \mbox{if } i=j\\
b:=\frac{\rho}{\sigma^2 \left(\rho(K-1)+1\right) (\rho-1)} & \mbox{otherwise}
\end{array} \right. ,$$
and $\delta<0$ is the mean for the alternative hypothesis. It is easy to see that the inverse $\Sigma^{-1}$ also has similar simple structure:

We show how to calculate the denominator of Eq. (\ref{Eq:Ti}) for this case in $O(K^2)$ complexity. Writing $ \mP(\vZ|\vh)$ explicitly, we get: 

$$ \mP(\vZ|\vh) = \frac{1}{(2\pi)^{K/2} |\Sigma|^{0.5}} \exp\left\{ -\frac{1}{2}\vZ^t \Sigma^{-1} \vZ\right\} \exp\left\{ -\frac{1}{2} \delta^2 \vh^t \Sigma^{-1} \vh\right\} \exp\left\{  \delta \vh^t \Sigma^{-1} \vZ\right\}.
$$
Examining the three terms in the exponents we observe that the first one is independent of $\vh,$ and the second one depends only on the number of non-zeroes in $\vh$. Assume this number is $k  = \vec 1^T \vh$, then we get:
$$
s_k := \vh^t \Sigma^{-1} \vh = a\cdot k  + b\cdot \left(k(k-1)\right)\;,\;\; k = \vec 1^T \vh. 
$$
where as defined above: 
$$ a=\frac{-1-\rho(K-2)}{\sigma^2 \left(\rho(K-1)+1\right) (\rho-1)}\;,\;\;b= \frac{\rho}{\sigma^2 \left(\rho(K-1)+1\right) (\rho-1)}.$$

Using these facts we can rewrite the denominator of Eq. (\ref{Eq:Ti}) as: 
\begin{eqnarray*}
\sum_{\vh} \mP(\vh) \mP(\vZ|\vh) &=& \frac{1}{(2\pi)^{K/2} |\Sigma|^{0.5}} \exp\left\{ -\frac{1}{2}\vZ^t \Sigma^{-1} \vZ\right\} \cdot  \\ && \sum_{k=0}^K \left[ \exp\left\{ -\frac{1}{2} \delta^2 s_k\right\} \sum_{\vh: k = \vec 1^T \vh} 
\exp\left\{  \delta \vh^t \Sigma^{-1} \vZ\right\} \right].
\end{eqnarray*}
Thus, the key is being able to calculate the $K+1$ sums in the last term efficiently. 

Using our $a,b$ notation we can write (assuming $ k = \vec 1^T \vh$): 
$$ \vh^t \Sigma^{-1} \vZ = b \cdot k  \cdot S_Z +(a-b) \vh^T \vZ\;,\;\;S_Z= \vec 1^T \vZ, $$ 
and we can take advantage of this representation to design a simple recursive algorithm to calculate the needed sums. Denote:
$$S(L,k) = \sum_{\vh \in \{0,1\}^L: k = \vec 1_L^T \vh} 
\exp\left\{  \delta \vh^t \Sigma_{L}^{-1} \vZ\right\} = \sum_{\vh \in \{0,1\}^L: k = \vec 1_L^T \vh} 
\exp\left\{  \delta \left(b \cdot k  \cdot S_Z +(a-b) \vh^T \vZ_L\right)\right\},$$
where $\Sigma_{L}^{-1}$ is the $L\times K$ matrix formed by taking the first $L$ rows of $\Sigma^{-1}$, 
and $\vZ_L$ denotes the first $L$ elements of $\vZ$. Thus, $S(K,k)$ is the $k^{th}$ sum in the needed calculation. Using this representation, it is easy to observe:
\begin{eqnarray*}
&& S(L,0)=1\;,\;0\leq L\leq K\;,\;\;\;\;S(0,k) = 0\;,\;0<k\leq K \\
&& S(L,k) = S(L-1,k)  + S(L-1,k-1) \exp\left\{\delta \left(b S_Z + (a-b) Z_L\right)\right\}\;,\;\;L>0,k>0.
\end{eqnarray*}
Each calculation of $S(L,k)$ given $S(L-1,\cdot)$ and precalculated $S_Z$ requires $O(1)$ operations, hence calculating $S(K,k)\;,\;k=0,\ldots K$ requires $O(K^2)$ calculations. 

Taking all of this together, we conclude that the complete calculation of:
\begin{eqnarray*}
\sum_{\vh} \mP(\vh) \mP(\vZ|\vh) &=& \frac{1}{(2\pi)^{K/2} |\Sigma|^{0.5}} \exp\left\{ -\frac{1}{2}\vZ^t \Sigma^{-1} \vZ\right\} \cdot  \sum_{k=0}^K \left[ \exp\left\{ -\frac{1}{2} \delta^2 s_k\right\} S(K,k) \right],
\end{eqnarray*}
can be done in $O(K^2),$ as all other calculations and summations except computing $S(K,k)$ are $O(K).$

For the calculation of the numerator of Eq. (\ref{Eq:Ti}), we can employ similar methodology, noting: 
\begin{eqnarray*}\sum_{\vh:h_i=0} \mP(\vh) \mP(\vZ|\vh) &=& \frac{1}{(2\pi)^{K/2} |\Sigma|^{0.5}} \exp\left\{ -\frac{1}{2}\vZ^t \Sigma^{-1} \vZ\right\} \cdot  \\ && \sum_{k=0}^{K-1} \left[ \exp\left\{ -\frac{1}{2} \delta^2 s_k\right\} \sum_{\vh: k = \vec 1^T \vh,h_i=0} 
\exp\left\{  \delta \vh^t \Sigma^{-1} \vZ\right\} \right].
\end{eqnarray*}
Hence we can follow the same steps as above to define $S^{(i)}(L,k),$ which evolves as $S(L,k),$ except when $L=i$:
$$ S^{(i)}(i,k) = S^{(i)}(i-1,k),$$
since we do not have to handle the case $h_i=1.$

Consequently, with a naive calculation of each $T_i$ numerator separately, we can calculate all locFDR values $T_1,\ldots,T_K$ in $O(K^3)$ complexity for this case. 
It seems plausible that an  approach for simultaneously generating $S^{(i)},i=1,\ldots,K$ could reduce the complexity to $O(K^2),$ but we do not currently have such an algorithm. 

In our experiments we were easily able to use this approach to calculate OMT policies for this setting with $K=1000.$  With the naive exponential calculation, any $K>50$ already becomes untenable.

\end{document}